\documentclass[final, 5p, times, 10pt, twocolumn, authoryear]{elsarticle}

\usepackage{multirow}
\usepackage{blindtext, graphicx}
\usepackage{url}
\usepackage{listings}
\usepackage{hyperref}
\usepackage{algorithm}
\usepackage{algorithmic}
\usepackage{amsmath}
\usepackage{amssymb}
\usepackage{graphicx}
\usepackage{booktabs}
\usepackage{graphicx}
\usepackage{caption}
\usepackage{subcaption}
\usepackage{amsfonts}
\usepackage{multirow}
\usepackage[dvipsnames]{xcolor}
\usepackage{fancyhdr}
\usepackage{enumitem}

\usepackage{graphicx}
\usepackage[font=small]{caption}
\usepackage[utf8]{inputenc}

\newtheorem{definition}{Definition}
\newtheorem{observation}{Observation}
\newtheorem{theorem}{Theorem}
\newtheorem{lemma}{Lemma}

\newcommand{\proof}[1]{\vspace{+10pt}\noindent\textbf{Proof of {#1}.}}

\begin{document}

\title{Operational Robustness of LLMs on Code Generation}
\author{Debalina Ghosh Paul}
\ead{19217422@brookes.ac.uk}

\author{Hong Zhu \corref{cor}}  
\ead{hzhu@brookes.ac.uk}
\cortext[cor]{Corresponding author.}

\author{Ian Bayley}
\ead{ibaybey@brookes.ac.uk}

\affiliation{
    orgnization = {School of Engineering, Computing and Mathematics, Oxford Brookes University},
    city = {Oxford},
    postcode = {OX3 0BP},
    country = {United Kingdom}
}

\begin{abstract}
It is now common practice in software development for large language models (LLMs) to be used to generate program code. It is desirable to evaluate the robustness of LLMs for this usage.  
This paper is concerned in particular with how sensitive LLMs are to variations in descriptions of the coding tasks. However, existing techniques for evaluating this robustness are unsuitable for code generation because the input data space of natural language descriptions is discrete.
To address this problem, we propose a robustness evaluation method called \emph{scenario domain analysis}, which aims to find the expected minimal change in the natural language descriptions of coding tasks that would cause the LLMs to produce incorrect outputs.
We have formally proved the theoretical properties of the method and also conducted extensive experiments to evaluate the robustness of four state-of-the-art art LLMs: \emph{Gemini-pro},  \emph{Codex},  \emph{Llamma2} and  \emph{Falcon 7B}, and have found that we are able to rank these with confidence from best to worst.  
Moreover, we have also studied how robustness varies in different scenarios, including the variations with the topic of the coding task and with the complexity of its sample solution, and found that robustness is lower for more complex tasks and also lower for more advanced topics, such as multi-threading and data structures.  
\end{abstract}


\pagestyle{fancy}
\lhead{Robustness of Code Generation, \today}
\rhead{D. G. Paul, H. Zhu and I. Bayley}

\maketitle



\section{Introduction}\label{sec:Introduction}


Robustness has been a hard problem and an active topic of research in the field of machine learning (ML) ever since \citet{szegedy2014}'s discovery of the universal existences of adversarial examples in deep neural networks (DNNs). It is a critical quality attribute of ML applications. 

In general terms, the robustness of an artificial intelligence (AI) system, as defined by \citet{iso24029-1-2023} Standard 24029, is the ability of the system to maintain its level of performance under any circumstances. Existing works on robustness of ML models can be classified into two types, according to the circumstances concerned: 
\begin{itemize}
\item \emph{Adversarial robustness} is used when an ML model, or its application system, is under attack from an adversary who intends to force the ML model to produce incorrect outputs. Robustness, in this context, means that the ML model performs as expected by resisting such adversarial attacks. It is closely related to security and has been an active research topic for the past decade; see \citep{khamaiseh2022} for a recent survey. 

\item \emph{Operational robustness} (or \emph{non-adversarial robustness}) is used when the operational conditions (or scenarios) vary naturally without malicious attacks. Here, robustness is the ability of the model to maintain its performance in a wide spectrum of scenarios. For example, an autonomous vehicle should drive safely not only in fine weather but also in rain, snow and fog.
\end{itemize}

This paper is concerned with the operational robustness of large language models (LLMs) as code generation tools, which turn natural language descriptions of requirements into program code. This is a widely used application of LLMs and there is overwhelming evidence that it can significantly improve the productivity of software engineers although concerns have been raised regarding the code quality of the output; see, for example, recent surveys by \citet{DORA2024} and \citet{Harding2025}. Operational robustness affects the usability of these code generation tools and therefore it needs to be understood.

There are two approaches to assessing the robustness of an ML model.

\begin{itemize}
\item The \emph{micro approach} focuses on the ML model's sensitivity to changes in the input, and evaluates the robustness by analysing the impact of changes on each input value. For example, if the ML model is the object detection model for an autonomous vehicle, a micro robustness property is the expected minimal number of pixels in an image to be changed to cause error in object detection. This approach is taken in the research on adversarial robustness.
\item The \emph{macro approach} assesses the stability in performance on various sub-domains of the input space, each of which represents a different scenario. This approach can be used for measuring operational robustness. In the autonomous vehicles example given above, a macro robustness property is that the accuracy of object detection remains high (e.g. over 90\%) in all weather conditions. 
\end{itemize}

Our previous work \citep{ghosh2024} found a significant decrease in the correctness of code produced by LLMs when the complexity of the coding task increases. Similarly, in \citep{CodeSmell2025}, we found that the presence of code smells, another indicator of code quality, varies significantly according to coding topics and complexity of the task. 
Both of these are macro robustness properties. However, far less is known about micro-level operational robustness. For example, can replacing one word with its synonym in a coding task description lead to a completely different program? This is the problem to be addressed in this paper. 

The remainder of the paper is organised as follows.


Section \ref{sec:RelatedWork} reviews related work on how to evaluate the robustness of ML models, discussing their merits and drawbacks, as well as their applicability to our research problem. From this, we will identify the challenges in robustness evaluation and outline our approach.
Section \ref{sec:ProposedMethod} presents and studies the theoretical framework of the proposed method as an example of micro approach to operational robustness, which is then applied to LLMs for code generation. 
Section \ref{sec:Implementation} describes the test system that implements the proposed method in the datamorphic testing methodology. 
Section \ref{sec:Experiment} reports the experiments on LLMs to demonstrate the effectiveness and efficiency of the proposed method. We will also study the best way to apply the proposed method and discuss the threats to the validity of the experiment. 
Finally, Section \ref{sec:Conclusion} concludes the paper with a summary of the main contribution of the paper, the implications of the results and the directions for future work.

\section{Related Work}\label{sec:RelatedWork}

In this section, we review the existing work on testing and evaluating the robustness of ML models, considering  adversarial robustness in Subsection \ref{sec:AdversarialRobustness}, and then operational robustness in Subsection \ref{sec:OperationalRobustness}. We will limit our scope to the testing and evaluation of robustness, rather than techniques for improving it. We will, on the other hand, broaden our scope from code generation to natural language processing (NLP) and image processing in general. 

\subsection{Adversarial Robustness}\label{sec:AdversarialRobustness}

Research on adversarial robustness has focused on the generation of adversarial examples (AEs) based on an attack model either explicitly or implicitly. An adversarial example is a constructed input slightly different from a natural input which causes the ML model to fail but has the same semantics as the natural input.

Two types of approaches exist for generating AEs: white-box and black-box. 

\subsubsection{White-Box Generation of AEs}

In white-box approaches, the attacker has full knowledge of the ML model, including its architecture and parameters, and utilises this knowledge to craft the AEs. 

The first white-box approach is perhaps the \emph{L-BFGS} method proposed by \citet{szegedy2014} and used in their study of the intriguing properties of DNNs. This approach was refined for the particular domain of computer vision and then adapted to NLP. We will discuss both of these areas in turn.

\citet{goodfellow2014} improved upon \emph{L-BFGS} by utilising the gradient of the ML model's loss function with respect to the input to give the \emph{Fast Gradient Sign Method} (FGSM) method. \citet{kurakin2018} in turn improved upon the success rate of FGSM by iterating on the gradients to give \emph{Basic Iterative Method} (BIM) for untargeted attacks and \emph{Iterative Consistency Loss Minimisation} (ICLM) for targeted attacks. \citet{moosavi2016} introduced \emph{DeepFool}, which seeks to push the input across the decision boundary with minimal perturbation. An improvement on minimal perturbation would be to modify just a few pixels rather than all and this is achieved by the \emph{Jacobian-based Saliency Map Attack} (JSMA) \citep{papernot2016}, which identifies the most impactful pixels and alters just them.

The methods above need to be run once for each test case in the benchmark to generate one AE and this can be time-consuming. In contrast, the \emph{Universal Adversarial Perturbation} (UAP) algorithm proposed by \citet{moosavi2017}, generates a single perturbation that when applied to different inputs generates different AEs. 

With the advance of defence mechanisms against adversarial attacks, AE generation techniques have also been developed to bypass these mechanisms. Examples include \citet{carlini2017}'s \emph{C\&W} attack, and \citet{madry2017}'s \emph{Projected Gradient Descent} (PGD) technique. 

As pointed out by \citet{yefet2020adversarial}, these algorithms cannot easily be adapted to NLP models, because although natural language sentences can be represented as vectors, not all vectors in the vector space will be meaningful inputs; an image with noise added to it is still an image but a valid English sentence may not be valid after the noise is added. 


However, at least three works have successfully adapted these algorithms to NLP in the sub-domain of sentiment analysis. For example, \citet{ebrahimi2018} adapted FGSM to \emph{HotFlip}, which works by identifying the most influential characters in a word based on the model's gradients, then flipping those characters to generate an AE. \emph{TextFool} proposed by \citet{liang2018} changes entire words rather than characters and thereby achieves more impactful modifications.  \emph{TextBugger} proposed by \citet{li2019textbugger} is a JSMA-based algorithm that manipulates specific words in the text by calculating the Jacobian matrix. It then applies subtle perturbations, like inserting spaces, deleting letters, replacing letters, or swapping adjacent letters to reduce the model's confidence in the correct prediction.

Finally, an example in the field of program analysis and more relevant to this paper is the \emph{Discrete Adversarial Manipulation of Program} (DAMP) method proposed by \citet{yefet2020adversarial}. DAMP generates its AEs by using the ML model's parameters to calculate the gradient towards the input that will cause the model to make incorrect predictions to guide the application of semantic preserving transformations on program code.  

We now consider black-box AE generation and we likewise discuss computer vision first and then NLP.

\subsubsection{Black-Box Generation of AEs for Computer Vision}

White-box attacks have a high success rate but black-box attacks generate AEs without requiring detailed knowledge of the ML model's internal structure or parameters \cite{khamaiseh2022}, information that is usually not available. Instead, the attacker can only see the output labels and associated confidence scores. 

Black-box attacks can be further categorised into \emph{transfer-based} and \emph{non-transfer-based} approaches. The former builds a surrogate model and apply white-box testing to it, while the latter does not, which is also known as \emph{query-based} in the literature. 

Transfer-based approaches are based on an intriguing property of DNNs discovered by \citet{szegedy2014} that an AE for one ML model is very likely to be an AE for another ML model of the same ML task even if the model is trained on different dataset and/or has different architecture and hyper-parameters. The most well-known examples of the transfer-based approach include \emph{UPSET}, which generates a universal perturbation for each target class, and \emph {ANGRI}, which produces input-specific perturbations. Both are presented by \citet{sarkar2017}. 

Non-transfer-based approaches do not build a surrogate model, and are further categorised into \emph{score-based}, which use confidence scores to craft AEs, and \emph{decision-based} which use only the decisions given by the model.

The first score-based algorithm, introduced by 
\citet{chen2017}, was \emph{Zeroth-Order Optimisation (ZOO)}, which performs hill climbing by performing numerous queries to the model in order to  approximate the gradient in confidence scores.
The Houdini algorithm developed by \citet{cisse2017houdini} maximises instead a separate measure known as task loss and requires far fewer queries.
The \emph{Simple Black-box Attack (SimBA)}, introduced by \citet{guo2019}, improves the query efficiency further by perturbing the input in random directions. Finally, the \emph{One Pixel Attack} proposed by \citet{su2019} is based on an evolutionary algorithm, which has the advantage of only targeting a single pixel.

The first decision-based algorithm was the \emph{Boundary Attack} algorithm by \citet{brendel2017}, which starts with a random input and finds a point on the classification boundary using binary search. Once on the boundary, it takes a number of random walks to find a point on the boundary closer to the original input. Note that to achieve a closeness comparable to existing white-box algorithms, millions of queries are required for each AE; this is because a large number of random walks along the boundary are required.

The great number of research efforts made to reduce this number include \citet{Brunner2019} and \citet{Liu2020}'s works which use biased samplings in place of uniform distribution to determine the random walk direction. Similarly, \citet{Maho2021}'s \emph{SurFree} algorithm directs the random walk by querying the neighbours of the boundary point to determine geometric information, and \citet{Li2021}'s \emph{Aha} algorithm combines this with historical queries.

Alternatively, walking direction can be viewed as an optimisation problem and random walks can be replaced by gradient-directed moves. Works of this kind include \citet{Cheng19}, \citet{Chen2020}'s \emph{HopSkipJump} algorithm, \citet{rahmati2020geoda}'s \emph{GeoDA} algorithm, \citet{Cheng2020Sign-OPT}'s \emph{Sign-OPT} algorithm, \citet{Li2020}'s \emph{QEBA} algorithm. The current state-of-the-art has reduced millions of queries to around a thousand queries for each image in the MNIST dataset. 

\subsubsection{Black-Box Generation of AEs for NLP}
The first work on black-box adversarial robustness for NLP is perhaps the work by \citet{jia2017}, which attempts to foil ML models for reading comprehension by constructing concatenative adversaries, variations of an original sentence with distracting additional phrases added.
The data set is the \emph{Stanford Question Answering Dataset} (SQuAD). Each test case has the form  $(p, q, a)$, where $p$ is the original text to be read by the NLP system $M$; $q$ is the question to be answered by $M$; and $a$ is the expected answer. A concatenative adversary is then a test case in the form $(p+s, q, a)$, where $s$ is a sequence of words concatenated to the end of the original text $p$ that does not contradict the correct answer.
The main method for producing $s$ is \emph{AddSent}, which produces a sentence that is syntactically identical but semantically different to the one in $p$ holding the answer. This is intended to fool NLP models that have learned to spot patterns but cannot recognise context in the way humans can. A variant \emph{AddSentDiverse} proposed by \citet{wang2018robust} adds $s$ potentially anywhere in the text $p$, not just at the end, enabling more diversity in the AEs. 
Another method \emph{AddAny} constructs $s$ as a sequence of random words, regardless of grammatically, drawn from the Brown corpus of 1000 most common English words and the words in $q$. A variant \emph{AddCommon} uses just words in the corpus. Yet another method \emph{AdddOneSent} takes $s$ to be a random human-approved sentence, allowing for what \citet{jia2017} call \emph{model-independent} adversaries.

Score-based black-box attacks modify various words or phrases in the input text and compare the output scores before and after the change. For example, the
\emph{DeepWordBug} algorithm proposed by \cite{gao2018blackbox}, targeting tasks like sentiment analysis and topic classification, ranks words based on their importance to the model’s prediction and applies character-level perturbations (e.g., insertion, deletion, substitution, or swapping). However, these changes often produce unnatural text. 

\emph{PWWS} (Probability Weighted Word Saliency) introduced by \citet{ren2019pwws} solves this problem by operating at the word level. It scores words based on their contribution to the 
prediction and replaces high-saliency words with their synonyms iteratively to reduce the model's confidence. 

Similar black-box algorithms have been developed for classification types of coding tasks. These include functionality classification, defect detection, clone detection etc. Examples of score-based algorithms include \emph{Metropolis-Hastings Modifier (MHM)}  \citep{zhang2020}, \emph{ALERT (Naturalness Aware Attack)} \citep{henkel2022}, and \emph{Style Transfer} \citep{alipour2021generalizability, li2022ropgen, pour2021search}. Perturbations, which must be semantics-preserving, include replacements of variable, function, or class names with alternatives and inserting code blocks. 
An example of a decision-based algorithm is \emph{WIR-Random} \citep{zeng2022} for clone detection, again with identifier-replacements as the perturbations.

\subsubsection{Metrics of Adversarial Robustness}
 
The most widely used adversarial robustness metric is perhaps the \emph{average relative distortion} of the attacks, i.e. the average distance between the original test case $x$ and the adversarial test case $x'=G(x,M)$ generated by the AE generation algorithm $G$. Formally, 

\begin{equation}
R_{G}(T,M) = \frac{1}{|T|} \sum_{x \in T}\frac{\|x,G(x,M) \|}{\|x\|} \label{equ:RMetric}
\end{equation}
where $T$ is the test dataset, $\| \cdot \|$ is a norm of the input data, $\|\cdot,\cdot\|$ is the distance metric induced from $\|\cdot\|$. 

Alternatives to this metric are possible. For example, \citet{fawzi2016robustness, fawzi2018analysis} calculated the \emph{average distortion} (i.e. not relative):
\begin{equation}
D_{G}(T,M) = \frac{1}{|T|} \sum_{x \in T} \|x,G(x,M) \| 
\end{equation}

Earlier, \citet{szegedy2014} 
used a variation of that metric that was restricted to the cases for which the ML model $M$ produces correct classifications; that is,

\begin{equation}
D^C_{G}(T,M) = \frac{1}{|T^C|} \sum_{x \in T^C}\|x,G(x,M) \| 
\end{equation}
where $T^C = \{x \in T | M(x) ~is~ correct \}$ is the subset of test dataset $T$ on which the ML model $M$ produces correct output. 

When the AE generation algorithm $G(x,M)$ always produces the minimum perturbation of the input $x$, then Equ. (\ref{equ:RMetric}) is a statistical estimation of the following theoretical metric for classification models proposed by \citet{moosavi2016}:
\begin{equation}
\rho(M) = \mathbb{E}_x  \frac{\Delta(x, M)}{\|x\|},
\end{equation}
where $\mathbb{E}_x$ denotes the expectation over the data distribution, and $\Delta(x,M)$ is the \emph{robustness of $M$ at the point $x$}: 
\begin{equation}
\Delta(x, M) = \min \{\|x', x\| ~| ~ M(x') \neq M(x)\}. 
 \end{equation} 

When input domain of $M$ is a vector space, a minimal $x'$ such that $M(x') \neq M(x)$ can be represented as $x'=x+r$. Thus, $r$ is called a \emph{minimal perturbation} of $x$. 

By definition, for every $x''$ such that $\|x, x''\| < \|r\|=\Delta(x, M)$, we have that $M(x'')=M(x)$. In other words, the size of minimal perturbation defines a safe zone of the ML model $M$.  Within the safe zone, the model $M$ will not fail. The robustness metric $\rho$ measures the expected size of the safe zones over the benchmark dataset. The greater the value of $\rho$, the larger the safe zone and the better the robustness. 

Since AE generation may fail, another important metric is the \emph{success rate} of attacks, which is formally defined as follows. 
\begin{equation}
S_{G}(T,M) = \frac{|\{x \in T ~|~ G(x,M) ~ \textit{is successful} \}|} {| T|}. 
\end{equation} 
where $G$ is a given AE generation algorithm and $T$ is the test dataset. The higher this success rate, the poorer the adversarial robustness.

It is worth noting that for both white-box and black-box algorithms, there is a trade-off between the success rate in finding an AE and the quality of that AE, in terms of distance from the original test case.

For example, in the FGSM algorithm (white-box), a greater step size is more likely to change the model's output but the distance would be greater. Similarly, 
for Boundary Attack (black box), conducting more random walks increases the probability that an AE will be found within the threshold specified.

Consequently, \citet{Li2021} used the \emph{area under the curve of success rate versus quality} to take both of these factors into account when comparing the performance of different AE generation algorithms (and the robustness of different models). Let $q=\varphi(s)$ denote the curve that plots attack quality $q$ against success rate $s$. The area under the curve (AUC) can be formulated as follows. 

\begin{equation}
AUC_{(SuccessRate,Quality)} = \int^1_0 \varphi(s) ~d s
\end{equation}

\subsection{Operational Robustness}\label{sec:OperationalRobustness}

Roughly speaking, in existing methods, the steps required to test and evaluate an ML model for operational robustness are as follows:

\begin{enumerate}
\item Identify the scenarios to be tested;
\item Construct a test dataset for each scenario;
\item Execute the model on each dataset and evaluate the performance of the ML model on each scenario;
\item Evaluate the robustness of the model based on its performance on different scenarios.  
\end{enumerate}

We now review the existing works by highlighting the differences in each of these steps. 


\subsubsection{Identification of Test Scenarios}

Three categories of methodologies to identify scenarios have been advanced in the literature. 
  
(a) \emph{Functional} approaches aim to cover all possible operational scenarios with respect to the functional requirements and the application domain. Typical work of this approach include Zhu \emph{et al}'s datamorphic test method \citep{zhu2019a}, Gao \emph{et al}'s model-driven test method \citep{gao2021model, gao2024ai, he2023framework, shu2024computer},  and Wotawa \emph{et al}'s ontology-based method \citep{kulmburg2025ontology, kulmburg2025state}. 
  
(b) \emph{Fault-based} approaches aims to evaluate the model's ability to maintain its level of performance when faults occur in the application system and in the operational environment where the input data comes from. Typical work in this category include \citet{hendrycks2019benchmarking}'s benchmarks \emph{ImageNet-C} and \emph{ImageNet-P} for evaluating image classifiers, and \citet{wozniak2022, wozniak2023, wozniak2024}'s design theory based method. 
  
(c) \emph{Error-based} approaches identify  erroneous operational scenarios by analysing the test cases on which the ML model fails. A typical example of this category is \citet{zhu2023scenario}'s work, where the model's performance is also improved on the scenarios identified. 


\subsubsection{Construction of Test Dataset}

There are two approaches in the literature for  constructing test datasets for different scenarios: \emph{profile-based} and \emph{transformational} approaches.

\noindent{(a) \emph{Profile-based approach}.} 

The profile-based approach for constructing test datasets for different scenarios starts with a large dataset taken from the real world. One technique is to filter the large dataset and check that each element meets the operational conditions. Alternatively, the large dataset could be divided up beforehand into a number of subsets, each representing an operational condition. Within the domain of code generations tasks, this has been done with the benchmarks APSS \citep{hendrycks2measuring}, for three difficulty levels, and with CoderEval \citep{yu2024codereval} for six levels of context dependency for the function to be written.

Alternatively, each test in the dataset can be associated with metadata to describe its operation conditions. This approach is more flexible as it allows many dimensions of operational conditions to be described. It was proposed by \citet{miah2024user} in a test dataset for R code generation and advanced further by \citet{ghosh2024} in the dataset ScenEval for Java code generation. 

which attempts to foil ML models for reading The profile based approaches can only be applied when data is naturally abundant and easy to collect. However, data in abnormal and hazardous operation conditions are usually rare and difficult to collect. Thus, in such scenarios, the transformational approach becomes necessary. 

\noindent{(b) \emph{Transformational approach}. }

The transformational approach takes a dataset $T$ for a different scenario and, for each input $x$ in $T$ applies a transformation $f$ to give an instance $x'=f(x)$ of the desired scenario. These instances form the dataset $T'=\{f(x)|x \in T\}$ for the new scenario. Typically, the existing dataset $T$ is for a typical and/or normal scenario.

A computer vision example is the work by \citet{hendrycks2019benchmarking}, who employed 75 image modification algorithms to construct benchmarks \emph{ImageNet-C} and \emph{ImageNet-P}. A subset of these algorithms were used by \citet{wozniak2024} to test an object detection model. Similarly, \citet{zhu2019a} used a Generative Adversarial Network (GAN) model to implement 13 transformations of face images, including changes of skin tone, age, hair colour, hair style, gender etc. to test face recognition models. Then, \citet{zhu2023scenario} used both a GAN model and coded algorithms to transform  images used to test the object detection ML model of an autonomous racing cars. 

The transformational approach has also been applied to code generation with semantics-preserving paraphrasing as the transformation. \citet{mastropaolo2023} investigated GitHub Copilot and compared two transformations: (i) a deep learning based paraphrasing tool PEGASUS \citep{zhang2020pegasus}, and (ii) translation from English to French and back again, a technique known as translation pivoting (TP). More recently, \citet{yan2024robustness} proposed a fully automatic technique called COCO in which the coding task is augmented with features extracted from the code generated for that task to form a paraphrase.
  

\subsubsection{Evaluation of ML Model's Performance}

Many performance metrics have been proposed and operational robustness can be evaluated based on any one of them.
%
%
It is out of the scope of this paper to review them in detail but readers are referred to the literature, e.g. \citep{sokolova2009systematic} and the \citep{iso24029-1-2023}. 


\subsubsection{Metrics for Evaluation of Operational Robustness}\label{sec:OperationalRobustnessMetrics}

As far as we know, all existing works on operational robustness are comparisons of the performances of the ML model on different operational scenarios. Typically, normal scenarios are compared with abnormal or atypical scenarios.

However, the performances on various scenarios are rarely aggregated into an overall measure of operational robustness. The only exception to this is \citet{hendrycks2019benchmarking}'s work on robustness of image recognition against various types of image corruptions. They introduced the notion of \emph{corruption robustness}, which is defined by the following equation: 
\begin{equation}
CR_C^D(M) = \mathbb{E}^M_{c \in C}[P_{x \in D} (M(c(x)) = l(x))]. \label{equ:OpR1}
\end{equation}
where $M$ is the classifier ML model under test, $D$ is a set of uncorrupted input data sampled according to a distribution $P$, and $C$ is a set of corruptions on the input data. For each $c \in C$ and $x \in D$, $c(x)$ is an image obtained by corrupting $x$ according to $c$, $l(x)$ is the true label of the image $x$, and $M(c(x))$ is the ML model $M$'s output label on the corrupted image $c(x)$. 

Note that the probability term denotes, as a measure of performance, the probability of classifier $M$ producing the correct output despite the corruption of the input by transformation $c$. Therefore, Equ. (\ref{equ:OpR1}) actually calculates the macro average of $M$'s accuracy over a number of scenarios. It can be generalised by replacing accuracy with any performance metric. 

In their experiment, the corruption robustness was evaluated and compared against a baseline model \emph{AlexNet} as follows. For each corruption $c \in C$, a test dataset was constructed by transforming uncorrupted images with 5 different levels of severity $i=1, \cdots, 5$. They calculated the error rate $E^M_{i,c}$ of a classifier $M$ when images are corrupted by $c$ at severity level $i$, then aggregated the error rates by taking their average. That is, 
\[ E^M_c =\sum_{i=1}^5{E^{M}_{i,c}}. \]

They further compared the corruption robustness of image recognition models with the baseline model \emph{AlexNet} using the following formula. 
\begin{equation}
CE^M = \frac{1}{|C|} \sum_{c \in C}{\left( E^M_{c}/E^{B}_{c} \right) } \label{equ:OpR2}
\end{equation}
when $B$ is the baseline model. 

They also proposed a \emph{relative robustness metric}, which was defined as follows. 
\begin{equation}
RelativeCE^M = \frac{1}{|C|}\sum_{c \in C} \left( \frac{\sum_{i=1}^5\left(E^M_{i,c} - E^M_{Clean}\right)} {\sum_{i=1}^5\left( E^{B}_{i,c} - E^{B}_{Clean}\right)} \right) \label{equ:OpR3}
\end{equation}
where $E^M_{Clean}$ is the error rate on the clean dataset $D$.

In summary, existing work on micro robustness is all about adversarial robustness, while work on operational robustness is all from the macro point of view. Testing and evaluating operational robustness from the micro point of view is still an open research problem. The next subsection identifies the main challenges and outlines our approach to overcome them. 

\subsection{The Challenges And Our Approach}

There are four reasons why existing methods and techniques are not applicable to measuring operational robustness of LLM code generation at the micro level.

\subsubsection{Generative Nature of the Task}

Robustness has only been explored for NLP on classification problems such as sentiment analysis and news classification tasks, and so on. For a generative task, the notion of failure is more complicated than a misclassification in a classification task. Code generation is such a generative task and it is hard to determine if a generated code is correct even if a reference solution is provided. Addressing this challenge, we do not require an absolute criterion to determine the correctness of the generated solution, but only a test oracle that can distinguish the solutions generated from two different inputs if they are not close enough. 

Assuming the benchmark of code generation comes with a reference solution and/or a set of test cases for the coding task, several approaches to automate such a test oracle have been proposed in the literature, which include (a) matching the string of tokens of generated code with the reference solution, (b) measuring the syntactical and structural similarity of the two solutions, (c) testing the functional equivalence between the generated and reference solutions on a set of test cases \citep{mastropaolo2023, yan2024robustness}.

The approach (c) above was further developed our previous work \citep{ghosh2024}, where we generated the required test cases from both the generated code and the reference solution. That approach will be adopted in this paper, too. 

\subsubsection{Discrete and Sparse Data Space}

Code generation is a task in the domain of natural language processing (NLP) and it relies heavily on a correct understanding of the natural language in which the code requirements are written. 

Not every string input is valid and meaningful text, however, so the input space is sparse, in contrast to that of image processing. It is also discrete rather than continuous so the metrics to measure the distance between two text inputs are far more complicated than those between two images, and it is not clear which metric is best. Furthermore, most algorithms for generating AEs depend on the data space being continuous.

We have therefore studied the effectiveness of many different metrics on texts for robustness evaluation. In response to the challenge of the data space being discrete and sparse, we only explore the meaningful points surrounding each input text rather than all of them. 

\subsubsection{Lack of an Attack Model}

Existing methods for evaluating micro robustness are based on the notion of adversarial inputs, which explicitly or implicitly relies on attack models. However, these attack models are meaningless for code generation task descriptions. The AEs generated by such algorithms do not represent input of the operational scenarios. 

Addressing this problem, we employ data augmentations to generate meaningful input data in a way similar the transformational approach to generate scenario-specific datasets. In particular, for code generation descriptions, we replace words in the natural language description with those that are semantically close by using a word embedding technique. 

\subsubsection{High Cost}

Leaving aside the problems above, existing AE generation algorithms are prohibitively expensive to apply to LLMs. White box algorithms are only applicable in the rare cases of open LLM models for which details are available and even in those situations, there would be so many parameters that it would take too long to generate test cases and / or it would require expensive powerful GPUs to process the data. Black box algorithms are even more expensive because they require so many invocations of the LLM to generate each AE. 

Addressing this challenge, we regard robustness testing and evaluation as a problem of exploratory testing \citep{zhu2022discovering} that aims at discovering the maximal safe distance within which a deviation from the original input would not cause the ML model to produce a different solution. We propose an automated exploratory testing method by integrating the ideas outlined above with a highly efficient exploration strategy that only explores meaningful data representing the operational scenario surrounding each input test case. Such meaningful data are generated by data augmentation using a word embedding technique and ordered by a distance metric. The search for the maximal safe deviation starts from the semantically closest meaningful data moving gradually to more distant ones, and stops when the difference from the original input is detected by the test oracle. Taking advantage of the sparse input space, it requires only a small number (less than 20) of black-box invocations of the LLM. 

The proposed method is presented in the next section. 

\section{The Proposed Method}\label{sec:ProposedMethod}

This section presents a micro method called \emph{scenario domain analysis} for scenario-based testing and evaluating the operational robustness of ML models. We first formally define a general theoretical framework and study its mathematical properties. The proofs of the theorems can be found in the Appendix. Then, we instantiate the theoretical framework with technical details to the problem of robustness of code generation. 

\subsection{Theoretical Framework} 

The basic idea of the proposed method is to explore the input space in the sub-domain of operational scenario surrounding a test case in a benchmark to discover the maximal size of the safe zone in which the ML model will not fail. The theoretical framework consists of two parts: (a) the definitions of the notion of "safe zone" and the metrics using the notion of safe zones to measure an ML model's robustness; (b) an algorithm or strategy to explore the sub-domain surrounding each test case in a benchmark. 

\subsubsection{Basic Concepts}

Generally speaking, a scenario in the use of an ML model $M$ is a specific type of situation defined by a specific operational condition in which the user invokes the ML model to achieve a specific purpose \citep{zhu2019a,zhu2023scenario, ghosh2024}. A scenario $s$, therefore, can be characterised by a sub-domain $D^s$ of the input space $D$ of the ML model $M$ that satisfies the operational condition that defines the type of usage situations. 

For example, driving in rain is a scenario in the operation of an autonomous vehicle. For the object detection ML model of the autonomous vehicle, this scenario can be characterised by the set of images of road traffic that satify the raining weather condition. Another example of a scenario in using a LLM as a coding tool is to develop a database component as the back-end of a SaaS application. This scenario can be characterised by the set of prompts that describe coding tasks to generate relational database SQL statements. 

As discussed in Section \ref{sec:RelatedWork}, micro operation robustness is concerned with how sensitive a ML model is to variations in the input data. Therefore, we assume that there is a distance metric (or similarity metric) $\|\cdot , \cdot \|$ defined on the input data space $D$ to measure the degree of variation. 

Moreover, we also assume that there is a test oracle that can determine if the ML model $M$ behaves differently between two input data $x$ and $y$. Considering $x$ as a deviation from a test case $t$, we say that ML model $M$ fails on $x$ if $M$ behaves differently according to the test oracle on $x$ from its behaviour on $t$, and we write $Fail^M_t(x)$. 

For a ML model to be robust in operation, one would expect that it resists small variations in its input; that is, it will behave in the same way when the input is changed slightly. Inspired by \citet{moosavi2016}'s notion of the minimal perturbation, we now define the notion of \emph{minimal deviation} in order to measure the operational robustness of an ML model $M$ in a scenario $s$ in terms of its capability to resist small variations in the input. 

\begin{definition}(Minimal distance to failure)\label{def:MinimalDistanceToFailure}

Let $t \in D$ and $D^s \subset D$ be the subset of the input space representing an operational scenario $s$.
 
An ML model $M$'s \emph{minimal distance from $t$ to failure} in the scenario domain $D^s$, denoted by $\Delta_{D^s}(t,M)$, is defined as follows. 

\begin{equation}
\Delta_{D^s}(t,M)=\min_{x \in D^s}\{\|x,t\| ~|~ Fails^M_t(x)\}. 
\end{equation}
The data $x\in D^s$ such that $\|x,t\|=\Delta_{D^s}(t,M)$ is called a \emph{minimal deviation} of $t$ in scenario $s$. 
\qed
\end{definition}

By the definition of $\Delta_{D^s}(t,M)$, we have that the following lemmas about the properties of minimal distance to failure. 

\begin{lemma}\label{lmm:lmmA}
$\forall x \in D^s. \left(\|x,t\| < \Delta_{D^s}(t,M) \Rightarrow \neg Fails^M_t(x)\right)$. \qed
\end{lemma}
\begin{lemma}\label{lmm:lmmB}
$\forall x \in D^s. \left(Fail^M_t(x) \Rightarrow \|x,t\| \geq \Delta_{D^s}(t,M)\right)$. \qed
\end{lemma}

In other words, $\Delta_{D^s}(t,M)$ defines a ``safe zone'' in which the ML model will behave the same way as $t$. Outside the safe zone, there is no guarantee that the ML model will behave the same way as expected. The data points in the safe zone can be considered to be ``safe''. The notion of \emph{safe} can be defined as follows. 

\begin{definition}(Safe)\label{def:Safe}

Let $t \in D$. We say that $x \in D^s$ is \emph{safe} for model $M$ w.r.t. to $t$, denoted by $Safe^M_t(x)$, if and only if (a) the model $M$ does not fail on $x$, and (b) $M$ does not fail on any points $x'$ in $D^s$ whose distance to $t$ is less than the distance from $x$ to $t$. Formally, 
\begin{eqnarray}
\lefteqn{Safe^M_t(x) = \neg Fails^M_t(x)~ \wedge }\nonumber \\
&& \forall x' \in D^s.( \| x', t\| < \| x,t\| \Rightarrow \neg Fails^M_t(x'))
\end{eqnarray}
\qed
\end{definition}

\begin{definition}(Maximal safe distance)\label{def:SafeDistance}

The \emph{maximal safe distance} from $t$ in a scenario domain $D^s$ for ML model $M$, denoted by $\nabla_{D^s}(t,M)$, is defined as follows. 
\begin{equation}
\nabla_{D^s}(t,M)=\max_{x \in D^s_t}\{\|x,t\| ~| Safe^M_t(x)\}. 
\end{equation}
\qed
\end{definition}


From Definition \ref{def:Safe} and \ref{def:SafeDistance}, it is easy to prove the following lemmas. Proofs are omitted for the sake of space. 

\begin{lemma}\label{lmm:lmmC}
$\forall x \in D^s.\left(\|x,t\| < \nabla_{D^s}(x,M) \Rightarrow Safe^M_t(x)\right)$. \qed
\end{lemma}

\begin{lemma}\label{lmm:lmmD}
$\forall x \in D^s.\left(Safe^M_t(x) \Rightarrow \|x,t\| \leq \nabla_{D^s}(t,M)\right)$. \qed
\end{lemma}

\begin{lemma}\label{lmm:lmmE}
$\forall x \in D^s.\left(\|x,t\| > \nabla_{D^s}(t,M) \Rightarrow \neg Safe^M_t(x)\right)$. \qed
\end{lemma}

From Definition \ref{def:MinimalDistanceToFailure} and \ref{def:SafeDistance}, it is easy to prove the following lemma about minimal distance to failure and maximal safe distance. 

\begin{lemma}\label{lmm:lmmF}
$\forall t \in D. \left(\nabla_{D^s}(t,M) \leq \Delta_{D^s}(t,M)\right)$. 
\qed
\end{lemma}

Note that for discrete spaces, there may be a gap between the maximal safe distance and minimal distance to failure. However, no matter whether or not the input space is discrete or continuous, we have the following property. 

\begin{theorem}\label{thm:NablaDelta}
There is no $x \in D^s$ such that $\nabla_{D^s}(t,M) < \|x, t\| < \Delta_{D^s}(t,M)$. \qed
\end{theorem}

Having defined the notions of minimal distance to failure and maximal safe distance, we can now revise the robustness metrics of adversarial robustness to define micro operational robustness. 

\begin{definition}(Operational Robustness in A Scenario)

The operational robustness in a scenario $s$, denoted by $\rho^{o}_s$, is defined by the following equation: 
\begin{equation}
\rho^{o}_s(M)= E_{x}\left(\Delta_{D^s}(x,M)\right) 
\label{equ:Rho-Metric}
\end{equation}

As an alternative to $\rho^{o}_s$, a new metric $\rho^*_s$ can also be defined with the following equation. 
\begin{equation}
\rho^*_s(M) = E_{x}\left(\nabla_{D^s}(x,M)\right) 
\label{equ:RhoStar-Metric}
\end{equation}
\qed
\end{definition}

\subsubsection{Exploration Strategy}

We regard robustness testing and evaluation as a form of exploratory testing aimed at discovering the minimal distance to failure and/or the maximal safe distance from a starting point $t$ in a benchmark $T$. As with the exploratory testing to find the boundary of classification models studied by \citet{zhu2022discovering}, there may be many different strategies to explore the sub-domain $D^s_\delta(t)$ with a starting point $t \in T$. Here, we propose a highly efficient strategy that requires a very small number of invocations of the ML model when the input space is discrete and sparse like the input domain of code generation. 

\begin{algorithm}[!h]
\scriptsize
\caption{Scenario Domain Analysis}
\label{alg:ScenarioDomainAnalysis}
\begin{algorithmic}[1]
\REQUIRE $ $ \\
    $T: \mathcal{P}(D)$; \textcolor{blue}{/* Test dataset */}\\
    $M: D \rightarrow C$; \textcolor{blue}{/* The ML model under test */}\\
    $GenMutants_n: D \rightarrow \mathcal{P}(D)$; \textcolor{blue}{/* Procedure to generate a set of mutants */}\\  
    $n$: Integer; \textcolor{blue}{/* To control the size of mutant set to be generated */}\\
    $\|\cdot,\cdot\|: D \times D \rightarrow \mathcal{R}^+$; \textcolor{blue}{/* Distance metric on $D$ */}\\      
	$Fail: C \times C \rightarrow Bool$; \textcolor{blue}{/* Test oracle; $Fail(c_1,c_2)=true$ means that $c_1$ and $c_2$ are different */}
\ENSURE  $ $\\
    $R^o: \mathcal{R}^+$; 
    $R^*: \mathcal{R}^+$;
\STATE $ListLS \leftarrow []$; \textcolor{blue}{/* To store the list of $LS(t)$*/}
\STATE $ListFF \leftarrow []$; \textcolor{blue}{/* To store the list of $FF(t)$*/}
\FOR{each $t$ in $T$}
	\STATE $ListMD' \leftarrow []$; \textcolor{blue}{/* To store the list of mutants and their distances for use after domain expansion */}
	\STATE $Mut_t \leftarrow GenMutants_n(t)$; \textcolor{blue}{/* Generate mutants from $t$*/}
	\STATE \texttt{Label1:} 
	\STATE $ListMD \leftarrow []$; \textcolor{blue}{/* To store the list of mutants and their distances to the seed $t$ */}
	\FOR{each $x$ in $Mut_t$}
    		\STATE $ListMD \leftarrow ListMD + \left<x,\|t,x\| \right>$;
	\ENDFOR
	\STATE $Sort(ListMD)$; \textcolor{blue}{/* in ascending order according to the distances  */}
	\STATE $Bool: AllSuccess \leftarrow$ \textbf{true}; 
	\STATE $LS \leftarrow \left<t,\|t,t\|\right>$;
	\STATE $C_t \leftarrow M(t)$; \textcolor{blue}{/* Test ML model $M$ on input $t$ to generate $C_t$ */}
	\FOR{each $\left<x, d_x \right>$ in $ListMD$ \textcolor{blue}{/* in ascending order of distances */}}
    		\STATE $C_x \leftarrow M(x)$ \textcolor{blue}{/* Test ML model $M$ on $x$ to generate $C_x$ */}  
   		\IF{$Fail(C_t,C_x)$}
        		\STATE $FF \leftarrow \left<x, d_x\right>$;      
        		\STATE $AllSuccess \leftarrow $  \textbf{false};      
        		\STATE \textbf{break}
   		\ELSE	
   			\STATE $LS \leftarrow \left<x, d_x\right>$;
    		\ENDIF
	\ENDFOR
	\IF{$AllSuccess$ \textcolor{blue}{/* No FF is found */}}
		\STATE \textcolor{blue}{// Expand the set of mutants}; 
		\STATE increase $n$; 
		\STATE $Mut_t \leftarrow GenMutants_n(t) - Mut_t$; 
		\STATE $ListMD' \leftarrow Merge(ListMD', ListMD)$; \textcolor{blue}{/* Ensure $ListMD'$ sorted in ascending order */}
    		\STATE \textbf{GoTo} \texttt{Label1}; \textcolor{blue}{/* To process the expanded set of mutants */}
    	\ELSE 
    		\STATE $LS \leftarrow BinarySearch(ListMD', FF)$ \textcolor{blue}{/* To find the maximal $\left<x, d_x\right>$ in $ListMD'$ such that  $d_x \leq FF.d$ /*} 
    		\STATE $ListLS \leftarrow ListLS + LS$;
    		\STATE $ListFF \leftarrow ListFF + FF$; 
	\ENDIF
\ENDFOR
\STATE $LSSum \leftarrow 0$; $FFSum \leftarrow 0$;
\FOR{each $\left<x,d_x\right>$ in $ListLS$}
    \STATE $LSSum \leftarrow LSSum + d_x$;
\ENDFOR
\FOR{each $\left<x,d_x\right>$ in $ListFF$} 
	\STATE $FFSum \leftarrow FFSum + d_x$;
\ENDFOR
\STATE $R^{o} \leftarrow FFSum / Length(T)$; 
\STATE $R^{*} \leftarrow LSSum / Length(T)$;
\RETURN $R^{o}, R^{*}$
\end{algorithmic}
\end{algorithm}

As shown in Algorithm \ref{alg:ScenarioDomainAnalysis}, in this strategy, the testing process consists of the following four steps. 

\begin{list}{Step 1.}
\item \emph{Generating sample sets of the scenario domain.}
\end{list}

For each test case $t$ in the benchmark test dataset $T$, a set $Mut(t)$ of test cases in the sub-domain $D^s$ of the scenario $s$ is generated systematically so that they are of varying similarities to $x$; see Line 5 of Algorithm \ref{alg:ScenarioDomainAnalysis}. 

The generated test cases in $Mut(t)$ are called \emph{mutants} of $t$, while the original test case $t \in T$ is called the \emph{seed} of the mutants. We require that the set $Mut(t)$ of mutants satisfy the validity and completeness conditions defined below.

\begin{definition}\emph{(Validity)}\label{def:validity}

A set $Mut(t)$ of mutants of a seed $t$ is valid, if every mutant correctly represents the operational scenario $s$ surrounding the seed $t$. Formally, $Mut(t) \subset D^s$. \qed
\end{definition}

\begin{definition}\emph{(Completeness)}\label{def:completeness}

A set $Mut(t)$ of mutants of seed $t$ is complete, if it contains all points of the sub-domain of the operational scenario $s$ that surround the seed $t$ within a certain distance $\delta_t$. Formally, there is a real number $\delta_t>0$ such that \[\forall x \in D^s. \|x,t\| \leq \delta_t \Rightarrow (x \in Mut(t)).\] 
The maximal $\delta_t$ that makes $Mut(t)$ complete is called $Mut(t)$'s \emph{coverage size} of the scenario $s$ surrounding $t$. \qed
\end{definition}

It is worth noting that, in practice, it may be difficult to formally prove that the set $Mut(t)$ of mutants are valid and complete and to find out its coverage size $\delta_t$, because the scenario sub-domain $D^s$ is usually not formally defined. Therefore, we require that the subset $Mut(t)$ of mutants are systematically generated, aiming to satisfying the validity and completeness conditions so that it provides a good characterisation of the scenario. 

\begin{list}{Step 2.}
\item \emph{Measuring the distances from mutants to the seed.} 
\end{list}

The mutants $x$ in the sample set $Mut(t)$ generated from a seed $t$ are measured for their distances to the seed $t$ using a distance metric $\|x,t\|$. These mutants are then sorted in ascending order according to their distances from the seed. See Line 7 - 11 of Algorthm \ref{alg:ScenarioDomainAnalysis}. 

For all $x,y \in D$, we write $x \prec_t y$ for $\|x,t\|< \|y,t\|$, $x \approx_t y$ for $\|x,t\| = \|y,t\|$, and $x \preccurlyeq_t y$ for $\|x,t\| \leq \|y,t\|$. 

So, the sorted set $Mut(t)$ of mutants can be represented in the form 
\[x_1 \preccurlyeq_t x_2 \preccurlyeq_t \cdots \preccurlyeq_t x_n, \]
where $n=|Mut(t)|$. 

In Section \ref{sec:ImpactOfSimilarityMetrics}, we will study various distance (and similarity) metrics on natural language texts and assess their suitability for use in robustness evaluations. 

\begin{list}{Step 3.}
\item \emph{Testing ML model and assess correctness of the output.} 
\end{list}    
    
The ML model is tested first on the seed test case $t \in T$ and then on the mutants $x \in Mut(t)$ in ascending order of the distance $\|x,t\|$ between the mutant $x$ and the seed $t$. That is from $x_1$ to $x_n$. The process stops when ML model fails on a mutant $a = x_k \in Mut(t)$. This is called the \emph{first failure} mutant and denoted by $FF(t)$. The immediate previous mutant $b=x_{k-1}$ before the first failure mutant $a$ is called the \emph{last successful} mutant and denoted by $LS(t)$. See Lines 14 -24 of Algorithm \ref{alg:ScenarioDomainAnalysis}. 

Note that it is possible that there are two or more mutants that are of the same distance to the seed $t$. We assume that the sorting in Step 2 puts such mutants in the sequence at random. Thus, the order that mutants are tested is not unique, if the testing process is executed more than once. In such cases, there may be more than one possible mutants that can be found as the first failure (and the last success) in different orderings of mutants. However, the first failure and the last success mutants always have the following properties. 

\begin{lemma}\label{lmm:lmm1}
For all $x \in Mut(t)$, $Fail^M_t(x) \Rightarrow (FF(t) \preccurlyeq_t x)$. \qed
\end{lemma}

\begin{lemma}\label{lmm:llm2}
$LS(t) \preccurlyeq_t FF(t)$. \qed
\end{lemma}

\begin{lemma}\label{lmm:lmm3}
There is no $x \in Mut(t)$ such that $LS(t) \prec_t x \prec_t FF(t)$. \qed
\end{lemma}
    
Let $\left<\alpha^*, \alpha^o \right>$ and $\left<\beta^*, \beta^o \right>$ be the pairs of last success and first failure points found in $Mut(t)$ and $Mut'(t)$, respectively. Then we have that the following theorem. 

\begin{theorem}\label{thm:Expansion2}
Assume that both $Mut(t)$ and $Mut'(t)$ are valid sets of the mutants of $t$. Then $Mut(t) \subseteq Mut'(t)$ implies that 
\begin{enumerate}
\item $\beta^o \preccurlyeq_t \alpha^o$. 
\item $\alpha^* \preccurlyeq_t \beta^*$, if $\|\alpha^*, t\| \leq \nabla_{D^s}(t,M)$.  \qed
\end{enumerate}
\end{theorem}

Note that the condition $\|\alpha^*, t\| \leq \nabla_{D^s}(t,M)$ of Theorem \ref{thm:Expansion2} Statement (2) normally holds when the set $Mut(t)$ is large enough. 

Theorem \ref{thm:Expansion2} proves that $LS(t)$ and $FF(t)$ moves monotonically with the expansion of the mutant set. The difference between their distances to the seed decreases with the increase of the set of mutants. Later, we will see (in Theorem \ref{thm:Ro} and \ref{thm:R*}) that they will eventually achieve accurately the value of $\nabla_{D^s}(t,M)$ and $\Delta_{D^s}(t,M)$, respectively. Therefore, the larger the set of mutants, the better accuracy of using $LS(t)$ and $FF(t)$ as approximations of $\nabla_{D^s}(t,M)$ and $\Delta_{D^s}(t,M)$. 

Here, we assume that the size of the mutant set $Mut(t)$ can be controlled, for example, by a parameter $n$ of the mutant generation function $GenMutants_n(t)$ so that the larger $n$ will results in a larger set of mutants. See Line 5, 27- 28 of Algorithm \ref{alg:ScenarioDomainAnalysis}. 

However, the larger the set of mutants the larger the cost of testing. This raises two questions in the practical application of the strategy: (a) how to determine the optimal size of the initial mutant set $Mut(t)$, and (b) how to determine the size of expansion when it is required. These problems are called the \emph{expansion strategy} problems. 

Generally speaking, according to Theorem \ref{thm:Expansion2}, the accuracy of using the last success and first failure points as approximations of $\nabla_{D^s}(t,M)$ and $\Delta_{D^s}(t,M)$ can be measured by the value $\|t, FF(t)\| - \|t, LS(t)\|$. It can also be used to determine if the set of mutants should be expanded. 

Moreover, there are two situations in which expansion could be considered. 

First, it is possible that the first failure mutant may not be found after testing all test cases in $Mut(t)$. This means that the coverage size  of $Mut(t)$ is not big enough. In that case, additional test cases should be generated to cover a larger zone surrounding $t$. 

Second, it is also possible the ML model $M$ fails on the first mutant in $Mut(x)$. Thus, there is no last successful mutant. This may also be caused by inadequate coverage of the set of mutants. In such a case, $\|t,FF(t)\|$ could be used to determine if the set of mutants should be expanded because by Theorem \ref{thm:Expansion2}, the true values of $\nabla_{D^s}(t,M)$ and $\Delta_{D^s}(t,M)$ will be less than or equal to $\|t,FF(t)\|$. 

When no first failure was found, the following theorems about domain expansion enable us to process domain expansion incrementally. Assume that  
\begin{equation}
\forall x \in Mut(t). \left(\neg Fail^M_t(x)\right), \label{eqn:thm3Proof2}
\end{equation}
and the set $Mut(t)$ of mutants of $t$ is expanded to a larger set $Mut'(t)$, i.e. 
\begin{equation}
Mut(t) \subset Mut'(t). \label{eqn:thm3Proof1} 
\end{equation}

\begin{theorem}\label{thm:Expansion}
Let $\left<\alpha^*, \alpha^o\right>$, and $\left<\beta^*,\beta^o\right>$ be the pairs of last success and first failure points found in the set $\left(Mut'(t)-Mut(t)\right)$ and $Mut'(t)$, respectively. Then, we have that
\begin{enumerate}
\item $\alpha^o \approx_t \beta^o$.
\item If $\alpha^* \approx_t \alpha^o$, we have that $\beta^* \preccurlyeq_t \alpha^*$ and there is no $x \in Mut'(t)$ such that $\beta^* \prec_t x \prec_t \alpha^*$. 
\item If $\alpha^* \prec_t \alpha^o$, we have that $\alpha^* \preccurlyeq_t \beta^*$. Moreover, 
\[\forall x \in Mut'(t).\left(\alpha^* \prec_t x \prec_t \beta^* \Rightarrow x \in Mut(t)\right). \] \qed
\end{enumerate} 
\end{theorem}

Note that, for statement (3) of Theorem \ref{thm:Expansion}, when $\|\alpha^*, t\| \geq Max\{\|x, t\|~|~ x \in Mut(t)\}$, there is no $x$ in $Mut(t)$ that satisfy the condition $\alpha^* \prec_t x \prec_t \beta^*$. 

Statement (1) of Theorem \ref{thm:Expansion} implies that, if the sub-domain is expanded from $Mut(t)$ to $Mut'(t)$ when no first failure is found on a set $Mut(t)$ of mutants, it is unnecessary to re-run the testing process on the whole expanded set $Mut'(t)$ of mutants. Instead, we can just run the testing process on the set $Mut'(t) - Mut(t)$ new mutants to find the first failure $\alpha^o$ and the last success $\alpha^*$ in the set $Mut'(t) - Mut(t)$. Using $\alpha^*$ and $\alpha^o$, we can work out the first failure $\beta^o$ and last success $\beta^*$ of $Mut'(t)$ as follows. 
\begin{eqnarray}
&& \beta^o = \alpha^o \label{eqn:corollary1}\\
&& \beta^* = \left\{ \begin{array}{ll}
\alpha^* & \textrm{if $\alpha^* \approx_t \alpha^o$} \\
Max\{x \in Mut(t) | x \preccurlyeq_t \alpha^o\} & \textrm{if $\alpha^* \prec_t \alpha^o$}
\end{array} \right. \label{eqn:corollary2}
\end{eqnarray}

The following Corollary of Theorem \ref{thm:Expansion} formally proves that the calculations by equations (\ref{eqn:corollary1}) and (\ref{eqn:corollary2}) are correct. 

\noindent\textbf{Corollary.}

Assume that $\forall x \in Mut(t).(\neg Fail^M_t(x))$, $Mut(t) \subset Mut'(t)$, and $\left<\alpha^*, \alpha^o\right>$ be the last success and first failure mutants found in testing on $Mut'(t) - Mut(t)$. 
There is a way to order the elements $\{x_1, \cdots, x_n\}$ in $Mut'(t)$ in the ascending order according to their distances to $t$ (i.e. $x_1 \preccurlyeq_t \cdots \preccurlyeq_t x_n$) such that the $FF(t) = \beta^o$ and the $LS(t) = \beta^*$, where $\beta^o$ and $\beta^*$ are defined by equations equations (\ref{eqn:corollary1}) and (\ref{eqn:corollary2}), respectively. \qed

Therefore, when a domain expansion happens, the testing process can be performed incrementally, which is much more efficient. The implementation of incremental processing can be found in Algorithm \ref{alg:ScenarioDomainAnalysis} in Lines 25 - 35. 

From Theorem \ref{thm:Expansion}, one might think that the smaller initial set of mutants the better for the efficiency of the algorithm. However, according to Theorem \ref{thm:Expansion2}, if the coverage size $\delta$ of the initial set of mutants is too small, even if you found the first failure point, it may result in an optimistic estimation of robustness. 

\begin{list}{Step 4.}
    \item \emph{Evaluating robustness.}
\end{list}   

Finally, the robustness of the ML model is calculated using the following formulas; see Line 37 - 43 of Algorithm \ref{alg:ScenarioDomainAnalysis}. 

\begin{equation}
R^o= \frac{1}{|T|} \sum_{t \in T}{\|t, FF(t) \|}  
\label{equ:RbyFF}
\end{equation}
\begin{equation}
R^*= \frac{1}{|T|} \sum_{t \in T}{\|t, LS(t) \|} 
\label{equ:RStarbyLS}
\end{equation}

Theorem \ref{thm:Ro} and \ref{thm:R*} below prove that, under the condition that the coverage size of the mutant set is large enough, $R^o$ and $R^*$ are accurate statistical estimations of the metrics $\rho^{o}$ and $\rho^*$, respectively. 

\begin{theorem}\label{thm:Ro}
If the set $Mut(t)$ is valid and complete with a coverage size $\delta_t \geq \|t,FF(t)\|$, then, we have that $\|t,FF(t)\| = \Delta_{D^s}(t,M)$. \qed
\end{theorem} 

Similarly, we have the following theorem about $LS(t)$. 

\begin{theorem}\label{thm:R*}
If the set $Mut(t)$ is valid and complete with a coverage size $\delta \geq \|t,FF(t)\|$, then, we have that $LS(t) \preccurlyeq_t \nabla_{D^s}(t,M)$, and there is no $x \in D^s$ such that $LS(t) \prec_t x \prec_t \nabla_{D^s}(t,M)$. \qed
\end{theorem}

Theorem \ref{thm:R*} means that $\|t,LS(t)\|$ is either accurate or the best conservative estimation of $\nabla_{D^s}(t,M)$. Moreover, the proof of the theorem shows that $\|t,LS(t)\|$ is not accurate (i.e. $\|t,LS(t)\| < \nabla_{D^s}(t,M)$) only if there is $\alpha \in D^s$ such that $\alpha \approx_t FF(t) \wedge  Safe^M_t(\alpha)$. In this case, $\nabla_{D^s}(t,M)=\Delta_{D^s}(t,M)$.

To ensure the accuracy of the results, by Theorem \ref{thm:Ro} and \ref{thm:R*}, we require that the set $Mut(t)$ is big enough to ensure that its coverage size $\delta_t$ is greater than the minimal distance to failure. 

It is worth noting that in practice, it may be difficult, even impossible, to find out accurately the coverage size $\delta_t$ of a set of mutants, because the scenario sub-domain $D^s$ is usually not formally defined. A practical solution is to set the parameter $n$ of mutant generation function $GenMutants_n(t)$ large enough so that $\delta_t > \|t,FF(t)\|$. 
    
It is worth noting two differences that this method has from the methods used to measure adversarial robustness. First, all mutants created are in the operational scenario concerned, whereas with adversarial robustness they can represent any scenario, or even be meaningless. Secondly, the mutants can be generated independently of the model under test, whereas almost all adversarial robustness algorithms depend on the model under test; see Section \ref{sec:AdversarialRobustness}. Thus, our method can provide a fair comparison of the ML models according to their operational robustness. 

In the following subsections, we give the technical details of the proposed method for the problem of measuring the robustness of program code generation. 

\subsection{Generating Sample Sets of the Scenario Domain}

The input data space for code generation consists of coding tasks described in natural language. The mutants are therefore paraphrases of the original descriptions of the coding tasks. 

To generate a large number of paraphrases of different distances from an original description of the coding task, we employ word embedding techniques, which provide a vector representation of the semantics and usages of words. These techniques have been intensively studied and widely used in NLP; two of the most popular are Word2Vec \citep{mikolov2013} and GloVe \citep{pennington2014}. 

\emph{Word2Vec}, developed by \citet{mikolov2013}, is a ML model that provides a high-dimensional vector for each word in such a way that words with similar meanings have vector representations close to each other in the vector metric space. It also allows for vector arithmetic operations applied to the vector representations of words to reveal their meanings. For example, $\left<king\right> - \left<man\right> + \left<woman\right> = \left<queen\right>$, where $\left<word\right>$ is the vector representation of the word.

\emph{GloVe}, introduced by \citet{pennington2014}, is a model obtained by an unsupervised ML algorithm that generates word embeddings by statistics of word co-occurrences in a text corpus. GloVe constructs a global word-word co-occurrence matrix, capturing the frequency with which words appear together in a given context. GloVe produces dense vector representations of words by factoring the matrix, ensuring that the resulting embeddings reflect both local and global statistical information about the words. 

We utilize GloVe, which is pre-trained, for creating paraphrases rather than Word2Vec because Word2Vec primarily focuses on local context but GloVe captures the relationships between words based on their global co-occurrence in the text. This allows it to facilitate the generation of paraphrases that are both contextually relevant and semantically meaningful. 

The paraphrase generation algorithm is based on the algorithm developed by \citet{zhu2020} to generate high order mutants via combinations of data mutation operators. Let's first introduce some notions and notations. 
		
Let $x$ be a sentence in natural language and $w$ be a word in $x$. The operation that replaces an occurrence of the word $w$ in $x$ by another word $w'$ is called the $ReplaceWord$ operation and denoted by $x[w/w']$. This operator is a \emph{data mutation operator} or \emph{datamorphism} \citep{zhu2020}.   
		
A paraphrase $x'= x[w/w']$ obtained by an application of the $ReplaceWord$ operator is a \emph{first order mutant} of $x$. 
		
Let $y$ be a $k$'th order mutant of $x$ obtained by replacing different word occurrences $w_1, \cdots, w_k$ with words $w'_1, \cdots, w'_k$ in $x$. A paraphrase $y'$ obtained by applying the $ReplaceWord$ operator on $y$, i.e. replacing a word $w_{k+1}$ in $y$ by a word $w'_{k+1}$, is an \emph{$(k+1)$'th order mutant}, if $w_{k+1} \neq w_i$ for all $i=1, \cdots, k$. 
		
The paraphrase generation algorithm consists of the following steps. 

\begin{enumerate}
    \item \emph{Sequentialising text}: The original text of the coding task description $x$ in natural language is converted into a linear sequence $Seq(x)=\left<w_1, \cdots, w_n \right>$ of words. 
    \item \emph{Retrieving neighbourhood}: Let $n>0$ be a given natural number. For each word $w$ in the sequence $Seq(x)$, the GloVe word embedding model is invoked to produce a set of $n$ words nearest to $w$ in the embedding vector space, where the set of words is denoted by $Neighbour_n(w)$. 
    \item \emph{Generating paraphrases}: Let $k>0$ be a given natural number and $k\leq L(x)$, where $L(x)$ is the length of the sequence $Seq(x)$. Our paraphrase generation algorithm first generates all first order mutants of the original coding task description by applying the $ReplaceWord$ operator $x[w/w']$, where $w \in Seq(x) \wedge w' \in Neighbour_n(w)$. It then generates all second order mutants by applying the $ReplaceWord$ operator on the first order mutants, and so on, until all $k$'th order mutants are generated. 
\end{enumerate}

For example, consider the coding task description: \textit{Write a Java program to replace a specified character with another character}. Table \ref{tab:Neighbourhood} gives the nearest 4 words to each of the words in the text of the coding task description, where the word with ranking equals to 1 is the nearest one according to GloVe. Table \ref{tab:Paraphrases} shows a small sample of paraphrases generated from the sentence with different values of parameter $k$ for the number of words replaced and $n$ for the highest ranking of the replacement words. As shown in this example, the semantic similarity between a paragraph and the original sentence cannot be easily determined by the parameters $n$ and $k$. Thus, it is important to order the set of paragraphs by measuring their similarity or distance to the original coding task description. 

\begin{table*}[t]
\caption{Examples of The Neighbourhoods of Words}\label{tab:Neighbourhood}
\scriptsize
\centering
\begin{tabular}{|c|l|l|l|l|l|l|l|l|l|}
\hline
\textbf{Rank ($n$)} &\textit{write} &\textit{a} &\textit{program} &\textit{to} &\textit{replace} &\textit{specified} &\textit{character} &\textit{with} &\textit{another}\\\hline
1 &writing &another &programs &take &replacing &applicable &characters &and &one\\\hline
2 &read &an &project &would &replacement &parameter &protagonist &while & a\\\hline
3 &publish &one &funding &instead &replaced &specifying &portrayed &made &an\\\hline
4 &notes &same &programme &could &replaces &depending &portrays &both &same\\\hline
\end{tabular} 
\end{table*}

\begin{table*}
\caption{Examples of Paraphrases}\label{tab:Paraphrases}
\centering
\scriptsize
\begin{tabular}{|l|c|c|}
\multicolumn{3}{l}{\textbf{Original}: \textit{Write a Java program to replace a specified character with another character.}} \\\hline
\textbf{Paraphrase} &$k$ &$n$\\\hline
\textit{Writing} a Java program to replace a specified character with another character.&1 &1\\\hline
Write a Java program to replace a \textit{applicable} character with another character. &1 &1 \\\hline
Write a Java program to replace a specified character with \textit{one} character. &1 &1\\\hline
Write a Java program to replace a specified character with \textit{a} character.&1 &2\\\hline
Write a Java program to replace a \textit{parameter} character with another character. &1 &2\\\hline
Write a Java program to replace \textit{one} specified character with another character. &1 &3\\\hline
Write a Java program to replace a specified \textit{protagonist} \textit{while} another character.&2 &2\\\hline
Write a Java program to replace \textit{an} \textit{applicable} character with another character. &2 &2 \\\hline
Write a Java program to replace a \textit{parameters} character with \textit{a} character.&2 &2\\\hline
Write \textit{an} Java program \textit{would} replace a \textit{parameter} character with another character.&3 &2\\\hline
\textit{Notes} a Java program to replace \textit{same} \textit{depending} character with another character. &3 &4\\\hline
\textit{Publish} a Java program to replace \textit{one} specified \textit{portrayed} \textit{made} another character. &4 &3\\
\hline 
\end{tabular} 
\end{table*}

\subsection{Measuring Distances from Mutants to the Seed}

At this step, each paraphrase generated in the previous step is measured for its distance to the original task description using a similarity metric between natural language sentences. 

In linguistics, the similarity between sentences can be measured not only based on the geometric distance, but also based on the lexical overlapping, structure and semantics as well. 

Text similarity metrics can be categorized into the following types. 
\begin{itemize}
\item \emph{Overlap-based metrics} measure similarity by counting common elements such as words or n-grams between texts and measuring the extent of n-gram overlap. Examples in this category include the metrics \emph{BLEU} \citep{papineni2002}, \emph{ROUGE} \citep{lin2004}, \emph{Meteor} \citep{banerjee2005}, and \emph{ChrF} \citep{popovic2015}.
\item \emph{Distance-based metrics}  embed texts in a vector space and quantify the similarity between texts by measuring either the spatial or angular distance between these vector representations in the embedding vector space. A typical example is the $L_2$ norm, which is most commonly used because it measures the \emph{Euclidean} distance and captures both the magnitude and direction of differences between vectors. Another example is the cosine similarity metric, which calculates the \emph{Cosine} of the angle of two vectors that represent the embeddings of the sentences in comparison. 
\item \emph{Edit distance metrics} quantify the number of editing operations required to transform one text into another. These operations typically include insertions, deletions, and substitutions. \emph{Levenshtein} distance is a widely-used edit distance metric, and it is classified into word-level and character-level types. The character-level distance calculates the minimum number of single-character edits needed to convert one text to another, while the word-level distance measures edits at the word level. 
\item \emph{Semantics-based metrics} employ ML models that map a pair of sentences to a similarity score according to the similarity in their meanings. For example, \emph{STSb-RoBERTa-base}\footnote{https://huggingface.co/cross-encoder/stsb-roberta-base} is such a ML model in the cross-encoder architecture RoBERTa \citep{liu2019roberta} fine-tuned on the Semantic Textual Similarity (STS) benchmark dataset to produce semantic similarity scores for pairs of text. In this paper, it is referred to as the \emph{RoBERTa} metric.

\end{itemize}


The metrics used and studied in this paper are: \emph{Euclidean}, \emph{Cosine}, \emph{BLEU}, \emph{ROUGE}, \emph{Meteor}, \emph{ChrF}, \emph{Levenshtein}, and \emph{RoBERTa}. 

\subsection{Testing ML Models and Assess the Correctness of Generated Code}

In this step, the ML model is tested first on the original coding task description $t$, and then on its paraphrases sequentially one by one in ascending order of distance from $t$. For each paraphrase $x$ of $t$, the algorithm queries the LLM under test with the paraphrase $x$ as the input, and checks if the code generated by the model is functionally equivalent to the code generated from the original task description $t$. 

The process terminates when it finds a paraphrase for which the generated code is not functionally equivalent to the original code. This paraphrase is referred to as the first failed paraphrase and denoted by $FF(t)$, and the paraphrase just before the first failed one is called the last successful paraphrase and denoted by $LS(t)$. These are used to calculate the $R^o$ and $R^*$ metrics, respectively, later in Step 4. This process is repeated on all code generation tasks in the test dataset. 

It is worth noting that the functional equivalence between two programs is not computable in general. We take a pragmatic approach by automatically generating test data from both programs and test the programs on the union of the test data; see \citep{ghosh2024} for more details. Formally, the predicate $Fails^M_t(x)$ is defined as follows. 

\[Fails^M_t(x) \Leftrightarrow \exists d \in TS(P) \cup TS(P'). (P(d) \neq P'(d)). \]
Here, $P=M(t)$ and $P'=M(x)$ are programs generated by ML model $M$, and $TS(P)$ and $TS(P')$ are the sets of functional test cases generated for $P$ and $P'$, respectively. 

\subsection{Evaluation of Robustness}

When the testing of the ML model $M$ finishes, we have the data $\{ \left(LS(t),FF(t)\right) ~|~ t \in T \}$. Then, the metrics $R^{o}$ and $R^*$ for model $M$ can be calculated using formula (\ref{equ:RbyFF}) and (\ref{equ:RStarbyLS}), respectively. 

The algorithms for testing and evaluation of the operational robustness of code generation is given in Algorithm \ref{alg:EvalRobustness}. 
 
\begin{algorithm}[ht]
\scriptsize
\caption{Robustness Evaluation of Code Generation}
\label{alg:EvalRobustness}
\begin{algorithmic}[1]
\REQUIRE $ $ \\
         $T$: Test dataset contains a list of coding tasks;\\
         $L_m$: The Large Language Model under test;\\
         $WE$: Word embedding model; \\  
         $\|\cdot,\cdot\|$: Distance metric on texts;\\      
         $n$: Integer; \textcolor{blue}{/* The number of nearest neighbours to used to generate paraphrases */}\\
         $k$: Integer; \textcolor{blue}{/* The number of words to be replaced to generate paraphrases */} \\
         $c_n$: Integer; \textcolor{blue}{/* The domain expansion size for $n$ */}\\
         $c_k$: Integer; \textcolor{blue}{/* The domain expansion size for $k$ */}
\ENSURE  $ $\\
        $\left<R^o, R^*\right>$: Robustness of the LLM $L_m$;
\STATE $ListLS \leftarrow []$; \textcolor{blue}{/* to store the list of $LS(x)$*/}
\STATE $ListFF \leftarrow []$; \textcolor{blue}{/* to store the list of $FF(x)$*/}
\FOR{each $t$ in $T$}
	\STATE $C_t \leftarrow L_m(t)$; \textcolor{blue}{/* Generate code $C_t$ using LLM $L_m$ for $t$ */}
	\STATE $T_t \leftarrow GenTest(C_t)$; \textcolor{blue}{/* Generate test cases for $C_t$ */}
	\STATE $P_t \leftarrow GenParaphrases(t, n, k, WE)$; \textcolor{blue}{/* Generate paraphrases from $t$*/}
	\STATE \texttt{Label1:} 
	\STATE $ListPD \leftarrow []$; \textcolor{blue}{/* The list of paraphrases and their distances to the original text */}
	\FOR{each $p$ in $P_t$}
    		\STATE $ListPD \leftarrow ListPD + \left<p,\|t,p\| \right>$;
	\ENDFOR
	\STATE $Sort(ListPD)$; \textcolor{blue}{/* Sort $ListPD$ according to the distances */}
	\STATE $AllPass \leftarrow$ \textbf{true}; 
	\STATE $LS \leftarrow \left<t,\|t,t\|\right>$;
	\FOR{each $\left<p, d\right>$ in $ListPD$ \textcolor{blue}{/* from shortest distance to the largest */}}
    		\STATE $C_p \leftarrow L_m(p)$ \textcolor{blue}{/* Generate code $C_p$ using LLM $L_m$ for $p$ */}  
   		\STATE $T_p \leftarrow GenerateTest(C_p)$; \textcolor{blue}{/* Generated test cases fro $C_p$*/}
   		\IF{$\forall x \in T_t \cup T_p. C_t(x) = C_p(x)$ \textcolor{blue}{/* $C_t$ and $C_p$ are functional equivalent */}}
   			\STATE $LS \leftarrow \left<p, d\right>$;
   		\ELSE	
        		\STATE $ListLS \leftarrow ListLS+LS$; 
        		\STATE $ListFF \leftarrow ListFF+\left<p, d\right>$;   
        		\STATE $AllPass \leftarrow $  \textbf{false};      
        		\STATE \textbf{break}
    		\ENDIF
	\ENDFOR
	\IF{$AllPass = $ \textbf{true} \textcolor{blue}{/* No FF is found */}}
	\STATE \textcolor{blue}{// Expand the set of paraphrases}; 
		\STATE $n \leftarrow n+c_n$; $k \leftarrow min\{k+c_k, |t|\}$; \textcolor{blue}{/* Increase the values of $n$ and $k$ */} 
		\STATE $P_t \leftarrow GenParaphrases(t, n, k, WE) - P_t$; \textcolor{blue}{/* Generate new paraphrases */}
    		\STATE \textbf{GoTo} \texttt{Label1}; \textcolor{blue}{/* To process the expanded set of paraphrases */}
	\ENDIF
\ENDFOR

\STATE $LSSum \leftarrow 0$; $FFSum \leftarrow 0$;
\FOR{each $\left<p,d\right>$ in $ListLS$}
    \STATE $LSSum \leftarrow LSSum + d$;
\ENDFOR
\FOR{each $\left<p,d\right>$ in $ListFF$} 
	\STATE $FFSum \leftarrow FFSum + d$;
\ENDFOR
\STATE $R^{o} \leftarrow FFSum / Length(T)$; 
\STATE $R^{*} \leftarrow LSSum / Length(T)$;
\RETURN $\left<R^{o}, R^{*}\right>$
\end{algorithmic}
\end{algorithm}

\section{Implementation of The Method}\label{sec:Implementation}

This section presents an implementation of our scenario domain analysis method to automate the testing and evaluation of the robustness of code generation models. It is based on the datamorphic testing methodology \citep{zhu2019a} and implemented based on the framework of test systems provided by the Morphy testing tool \citep{zhu2019}. 

\subsection{Datamorphic Test Methodology}

The datamorphic testing methodology proposed by \citet{zhu2019a,zhu2019} views software testing as a systems engineering problem, focusing on the development and implementation of a test system that automates testing processes, manages test resources efficiently, and co-evolves the test system with the software under test. 

In datamorphic testing methodology, a test system consists of two kinds of artefacts: (a) \emph{test entities}, such as test data, objects, the software under test, and test results, etc., and (b) \emph{test morphisms}, which are operations that perform test activities by using, transforming, and producing various types of test entities. 

The tool Morphy \citep{zhu2019} is designed to support datamorphic testing. It is utilised here to perform robustness testing and evaluation and to conduct the experiments. Morphy provides a framework for the implementation of test systems in Java. It consists of a set of classes that defines the general structure of test entities, including test cases and test suites, and also a set of metadata classes, so that a test system is conceptually a class with attributes that represent test entities and methods that represent test morphisms. User-defined test entities and morphisms are annotated with metadata so that they can be recognized and integrated with the automated tools, test strategies, and management functionality provided by Morphy, which recognises the following types of test morphisms. 

\begin{itemize}
\item \emph{Seed maker}: Generates test cases from other types of entities.
\item \emph{Datamorphism}: Transforms test cases.
\item \emph{Metamorphism}: Verifies the correctness of test cases and returns a Boolean value. 
\item \emph{Test set filter}: Adds or removes test cases from a test set.
\item \emph{Test set metric}: Maps a test set to a real value, such as test adequacy.
\item \emph{Test case filter}: Maps a test case to a Boolean value to determine whether it should be retained in the test set.
\item \emph{Test case metric}: Assigns a real value to test cases, such as their complexity. 
\item \emph{Analyser}: Analyses the test set and generates a test report.
\item \emph{Executer}: Runs the program under test using input data from the test case and captures the program's output. 
\end{itemize}

In addition to managing test entities, Morphy offers test automation at the following three levels. 

\begin{itemize}
\item \emph{Action}: Automates each test activity via an invocation of the corresponding test morphism and/or Morphy's predefined test management function.
\item \emph{Strategy}: Combines test morphisms using an algorithm that takes test entities and morphisms as parameters.
\item \emph{Process}: Executes actions and strategies according to an editable test script, which can be obtained by recording the user's interaction with Morphy through its GUI. 
\end{itemize}

\subsection{Structure of The Test System}

The implementation of the proposed method is built on top of the test system developed for the ScenEval benchmark \citep{ghosh2024} to test LLM's capability for code generation. It defines:  (a) test entities like test cases, test suites, and (b) test morphisms. In addition to the test morphisms implemented in \citep{ghosh2024}, a set of new datamorphisms, analysers and executors are implemented. The structure of the test system is shown in Figure \ref{fig:StructurOfTestSystem}.  

\begin{figure*}[h]
\centering
\includegraphics[width=16cm]{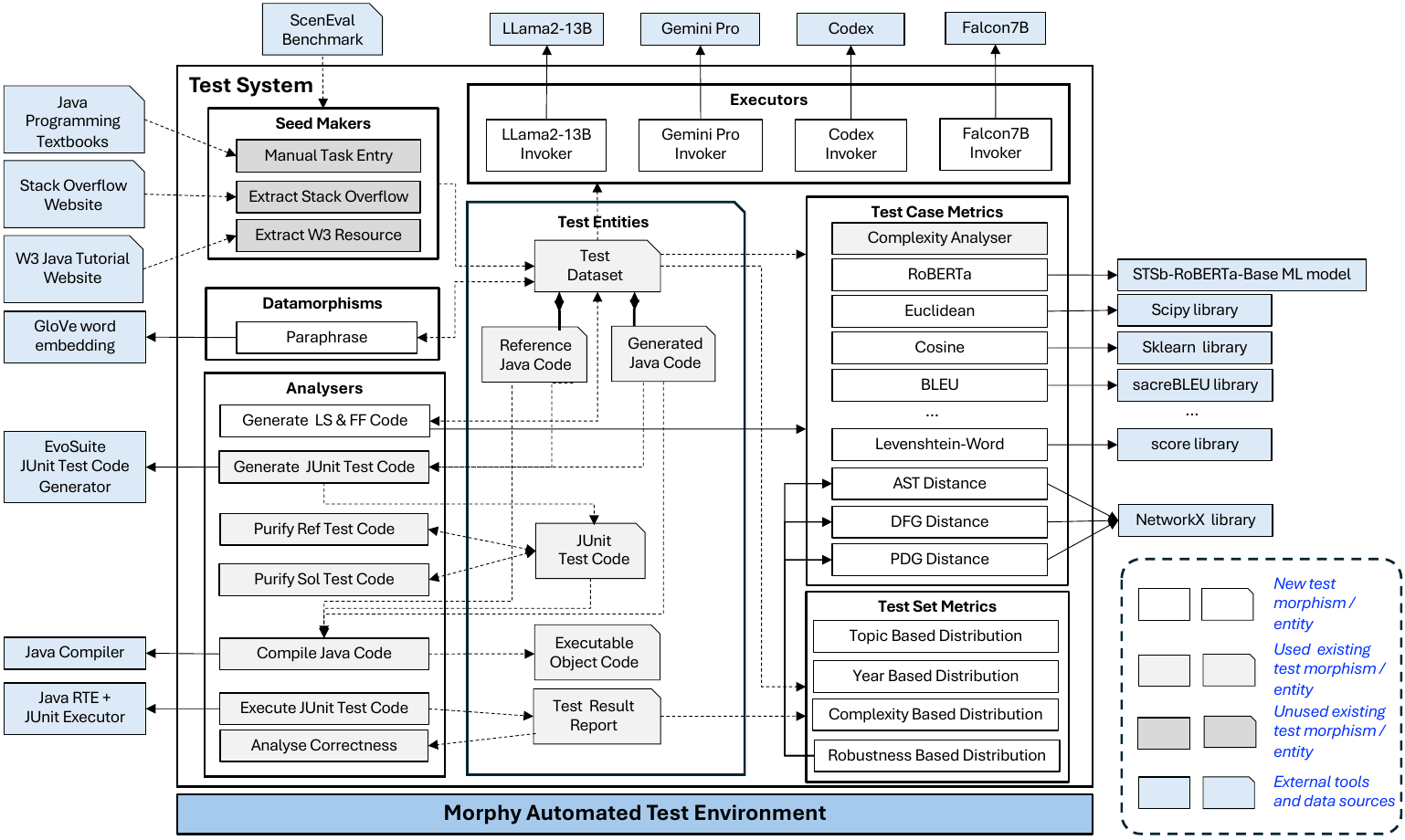} 
\caption{Structure of Test System for Evaluation of Robustness}
\label{fig:StructurOfTestSystem}
\end{figure*}

The following describes the test morphisms that have been implemented. 

\subsubsection{Datamorphisms}

A datamorphism called \emph{Paraphrase} implements the paraphrase algorithms described in Section \ref{sec:ProposedMethod}. It transforms each seed test case into a set of mutants that each contains a paraphrased description of the coding task. The datamorphism has two parameters: (a) $n$, the number of nearest words to be retrieved for each word, and (b) $k$, which is the highest order of mutants to be generated; this controls the maximal number of words to be replaced simultaneously. 

\subsubsection{Test case metrics}
Each similarity metric on texts is implemented as a test morphism that when applied to a mutant, returns the distance between the mutant and its seed. Table \ref{tab:ImplOfSimilarityMetrics} gives information about how each of them are implemented. 

\begin{table}[htbp]
	\caption{Implementation of Similarity Metrics}
	\label{tab:ImplOfSimilarityMetrics}
	\begin{scriptsize}
	\centering
	\begin{tabular}{|p{1.5cm}|p{6.5cm}|}
		\hline 
		\textbf{Metric} & \textbf{Implementation} \\ \hline
		RoBERTa & Query the STSb-RoBERTa-base ML model installed on local machine, which is downloaded from Hugging Face \\ \hline
		Euclidean & Invoke the \texttt{euclidean} function of the scipy.spatial.distance module \\ \hline
		Cosine & Invoke the \texttt{cosine\_similarity} function of the sklearn.metrics.pairwise module \\ \hline
		BLEU & Invoke the \texttt{corpus\_bleu} function in the sacreBLEU library, with parameter $n$ = 4 \\ \hline
		ROUGE-N & Use the \texttt{RougeScorer} class of the rouge\_score library \\ \hline
		ROUGE-L & Use the \texttt{RougeScorer} class of the rouge\_score library \\ \hline
		Meteor & Invoke the \texttt{meteor\_score} function of the meteor\_score library \\ \hline
		ChrF & Invoke the \texttt{sentence\_chrf} function from the sacrebleu library \\ \hline
		Levenshtein-Char & Invoke the \texttt{distance} function of the Levenshtein library with character as the unit \\ \hline
		Levenshtein-Word & Invoke the \texttt{distance} function of the Levenshtein library with word as the unit \\ \hline
	\end{tabular}
	\end{scriptsize}
\end{table}

\subsubsection{Test Executors}

Four test executors have been implemented, one for each of the LLM models GeminiPro, Codex, Falcon7B and LLama2-13B. Each executor submits a query to the LLM through API invocation, receives the responses from the LLM and then saves the code in the response to a local file. 


\subsubsection{Test Result Analysers}

This type of test morphism analyses the test results on each coding task and its paraphrases. Two types of analysers have been implemented.  

The first is \texttt{Robustness(DM: DistanceMetric)}, which calculates the robustness $R^{p}$ and $R^*$ of the LLMs according to the formulas (\ref{equ:RbyFF}) and (\ref{equ:RStarbyLS}). It takes as its parameter a distance metric and compares the LLMs by plotting the average distances from the seed to the first failed mutant and to the last successful mutant. Here, the test morphism invokes the distance metrics implemented as the test case metrics.

The second group of analysers compare the code generated from the paraphrases of $FF(x)$ and $LS(x)$ with the code generated from the original description. 
They calculate the editing distances between the program source codes using metrics defined on various models of program code structure, including: 


\begin{itemize}
    \item \emph{AST (Abstract Syntax Tree)}: Represents the syntactic structure of the source code in a hierarchical tree format, showing the grammatical structure of the code. The implementation uses the \texttt{javalang} library to create the AST and the \texttt{zss} module to compute the minimum AST edit distance between two programs.
  \item \emph{DFG (Data Flow Graph)}: Represents the flow of data between operations in the program, capturing how data is processed and transformed. The DFGs were generated using the \texttt{Soot} framework, and the graph edit distance was computed using the \texttt{NetworkX} library. 
  \item \emph{PDG (Program Dependency Graph)}: Depicts the dependencies between different parts of a program, including control and data dependencies. As with DFG, the PDGs were generated using the \texttt{Soot} framework, and the graph edit distance was computed using the \texttt{NetworkX} library. 
\end{itemize}

\section{Experiments}\label{sec:Experiment}

This section reports the experiments with the proposed method using the implemented test system in attempt to answer the following research questions. 

\begin{enumerate}
\item \emph{RQ1 (Effectiveness)} Is the proposed method capable of providing meaningful results that distinguish LLMs according to their robustness in code generation?
\item \emph{RQ2 (Efficiency)} Is the proposed method efficient enough for practical use?
\item \emph{RQ3 (Selection of Distance Metric)} Which similarity/distance metric is the best to use for robustness evaluation? 
\end{enumerate}

\subsection{Design of The Experiments}\label{sec:designOfExperiment}


Four LLMs capable of code generation were chosen as the subject in the experiment for their availability and low cost of access as well as their wide uses in practice at the time when the experiments started. The LLMs were \emph{Llama2-13B}, \emph{Codex}, \emph{Gemini Pro 1.0}, and \emph{Falcon}. 

The experiments were conducted with Morphy running on a desktop computer accessing LLM models running on the cloud through API invocations. Specifically, GeminiPro was executed via its API provided by Google. Falcon and Llama2 were deployed on Google Colab with GPU support, and Codex was accessed through GitHub Copilot. Figure \ref{fig:fig4_1} illustrates the configuration of the experimental platform.

\begin{figure}[h]
\centering
\includegraphics[width=7cm]{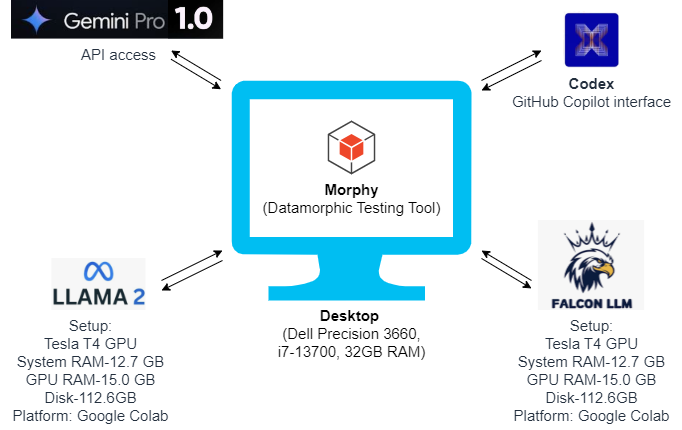}  
\caption{The Experiment Platform}
\label{fig:fig4_1}
\end{figure}


We formed a test dataset $D_1$ by selecting 900 coding tasks at random from the ScenEval benchmark \citep{ghosh2024}. 
Table \ref{tab:StatisticsOfTestDataset} gives the statistical characteristics of the test dataset $D_1$. 

\begin{table}[htbp]
\centering
\caption{Statistical Characteristics of the Test Dataset $D_1$}
\label{tab:StatisticsOfTestDataset}
\begin{scriptsize}
\begin{tabular}{|l|c|c|}
\hline
\textbf{Feature} & \textbf{\#Textbook Tasks} & \textbf{\#Real Tasks}\\ \hline\hline
Number of Topics Covered & 9 &7\\ \hline
Average Num of Tasks in A Topic & 66.67 & 42.85\\ \hline
Max Num of Tasks in A Topic & 156 & 61\\ \hline
Min Num of Tasks in A Topic & 8 & 5\\ \hline\hline
Average Complexity & 3.448 &3.200\\ \hline
Max Complexity & 13 &6\\ \hline
Min Complexity & 1 &1\\ \hline\hline
Average Length of Input & 16.800 &25.344 \\\hline
Max Length of Input & 35 &31 \\ \hline
Min Length of Input & 11 & 18\\ \hline\hline
Total Number of Tasks & 600 &300\\ \hline\hline
\end{tabular}
\end{scriptsize}
\end{table}

The testing was conducted in the process described in Section \ref{sec:ProposedMethod} for each LLM and repeated on each of the similarity metrics used to measure the distance between the paraphrases and the original coding task description. The experiment data and their analysis are reported in the following subsections to answer the research questions. 

\subsection{Effectiveness}

To answer the research question \emph{RQ1}, we first apply the method on the whole test dataset to evaluate and compare the robustness of the subject LLMs to demonstrate that the proposed method can distinguish LLMs on their robustness of code generation. Then, we apply the method to various scenarios of using LLMs in code generation to demonstrate that each LLM can be distinguished on their robustness in different scenarios. We also provide a comprehensible interpretation of the results. 

\subsubsection{Comparing LLMs on Overall Robustness} 

To compare the robustness of the LLMs, we apply the $R^o$ and $R^*$ metrics on the whole dataset of the experiment. The experiment results are given in Table \ref{tab:combined_robustness}. 

\begin{table*}[h]
\centering
\caption{Robustness Measured by $R^o$ and $R^*$ with Different Similarity Metrics on Paraphrases}\label{tab:combined_robustness}
\begin{scriptsize}
\begin{tabular}{|l|c|c|c|c|c|c|c|c|} 
\hline
\multirow{2}{*}{\textbf{Similarity Metric}} 
& \multicolumn{2}{c|}{\textbf{Gemini-Pro}} 
& \multicolumn{2}{c|}{\textbf{Codex}} 
& \multicolumn{2}{c|}{\textbf{Llama2}} 
& \multicolumn{2}{c|}{\textbf{Falcon}} \\ \cline{2-9} 
& $R^o$ & $R^*$ & $R^o$ & $R^*$ & $R^o$ &$R^*$ & $R^o$ & $R^*$ \\ \hline
RoBERTa &1.0365&1.1117&1.0879&1.1432&1.1945&1.2478&1.2046&1.2550\\ \hline
Euclidean &0.5307&0.4991&0.5103&0.4814&0.4021&0.3748&0.4065&0.3781\\ \hline
Cosine Similarity&0.9291&0.9302&0.9301&0.9381&0.9801&0.9811&0.9703&0.9712\\ \hline
BLEU &0.4273&0.4565&0.4164&0.4388&0.5336&0.5628&0.5295&0.5583\\ \hline
ROUGE-N &0.6224&0.6428&0.6354&0.6587&0.7414&0.7662&0.7381&0.7606\\ \hline
ROUGE-L &0.6950&0.7130&0.7066&0.7235&0.8104&0.8277&0.8085&0.8242\\ \hline
METEOR &0.6552&0.6785&0.6692&0.6925&0.7753&0.7995&0.7718&0.7939\\ \hline
ChrF &66.9836&68.8565&68.0947&69.8502&68.6146&70.2630&68.3506&70.0672\\ \hline
Levenshtein-Word &1.1343&1.0233&1.1288&1.0151&1.1260&1.0096&1.1248&1.0551\\ \hline
Levenshtein-Char &11.5151&10.1133&10.5876&9.4900&10.1340&9.1700&10.3320&9.3000\\ \hline
\end{tabular}
\end{scriptsize}
\end{table*}

Figure \ref{fig3}(a) and (b) depict the robustness of the LLMs under two similarity metrics RoBERTa and Euclidean, respectively, where the red lines are the average $R^{o}$ scores calculated from the first failed solutions; the black lines are the average $R^*$ scores calculated from the last successful solutions. 

\begin{figure}[htb]
\centering
\includegraphics[width=8cm]{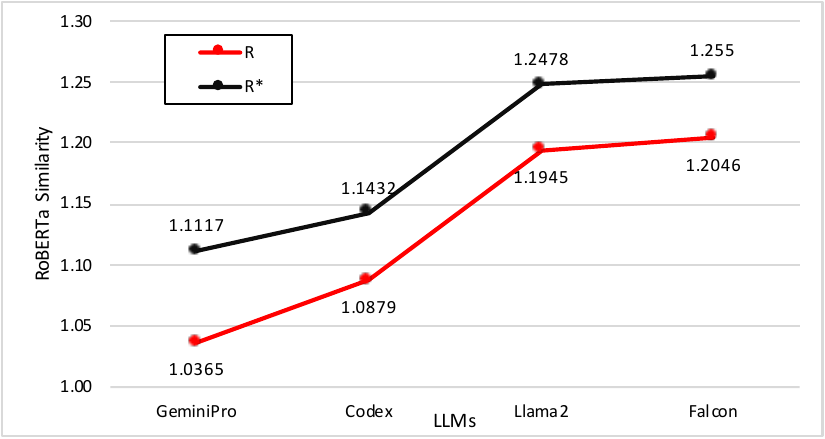} \\ 
\scriptsize{(a) With RoBERTa Similarity} \\
\includegraphics[width=8cm]{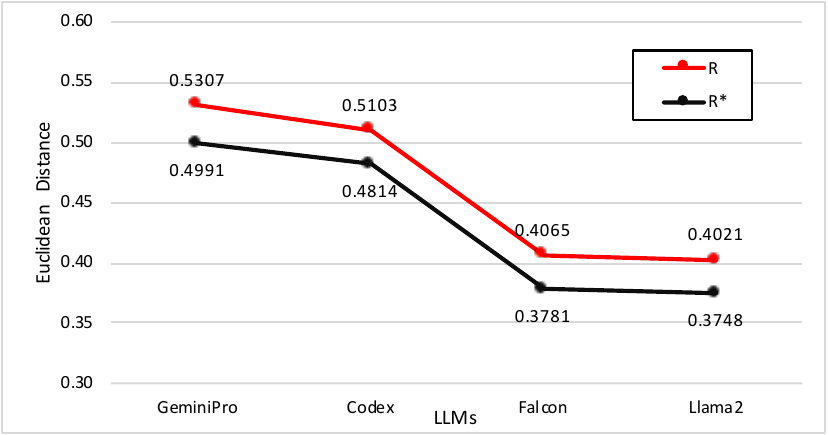} \\ 
\scriptsize{(b) With Euclidean Distance}
\caption{Comparison of LLMs on Robustness of Code Generation}
\label{fig3}
\end{figure}

Note that, when using Euclidean distance between two texts as a measure of similarity, a lower the score means the two texts are more similar to each other so a smaller  $R^{o}$ ($R^*$) value means a poorer robustness. In contrast, when using RoBERTa to measure similarity, the higher score means the text are more similar to each other so poorer robustness is indicated by a \textit{higher} $R^o$ or $R^*$ value because it means that the code generation model fails on more similar inputs. 

From the experimental data, three observations can be made. 

\begin{observation} 
The proposed domain analysis method for robustness testing and evaluation is capable of distinguishing ML models on their robustness in code generation regardless of the similarity metrics used. 
\end{observation} 

\begin{observation} 
The robustness metrics $R^o$ and $R^*$ are consistent across all similarity metrics. 
\end{observation} 

\begin{observation} \label{obs:RobustnessDependsOnSimilarity}
The results of comparing the robustness of ML models depends on the similarity metrics used. 
\end{observation} 

In particular, the data shows that Gemini-Pro is the best on robustness as it scores highest on $R^o$ with almost all similarity metrics except the RoberTa metric. Codex is the second best, which has the highest scores on only one similarity metric (RoberTa), but scored second on all others. 

Observation \ref{obs:RobustnessDependsOnSimilarity} above raises the research question: which similarity metric is best for robustness testing and evaluation? This will be investigated in Subsection \ref{sec:ImpactOfSimilarityMetrics}. 

\subsubsection{Comparing Robustness in Various Scenarios}

In order to apply the proposed method to various scenarios for which LLMs are employed for code generation, the dataset $D_1$ is split into a number of subsets such that each subset contains the coding tasks of the scenario. 
These scenarios are grouped by the coding task's complexity and by the coding topic. 

The complexity of a coding task is measured by the cyclomatic complexity of the reference solution provided by the ScenEval benchmark. A total of 6 subsets were formed with complexity ranging from 1 to 6. Then the model's robustness was evaluated using $R^o$ and $R^*$ metrics on each subset the same way as for the whole dataset. 

\begin{figure}[h]
\centering
\includegraphics[width=8cm]{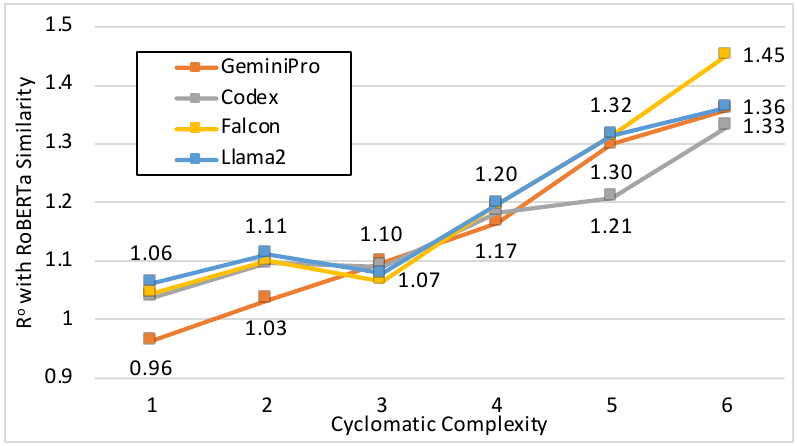} \\ 
 \scriptsize{(a) RoBERTa Similarity} \\
\includegraphics[width=8cm]{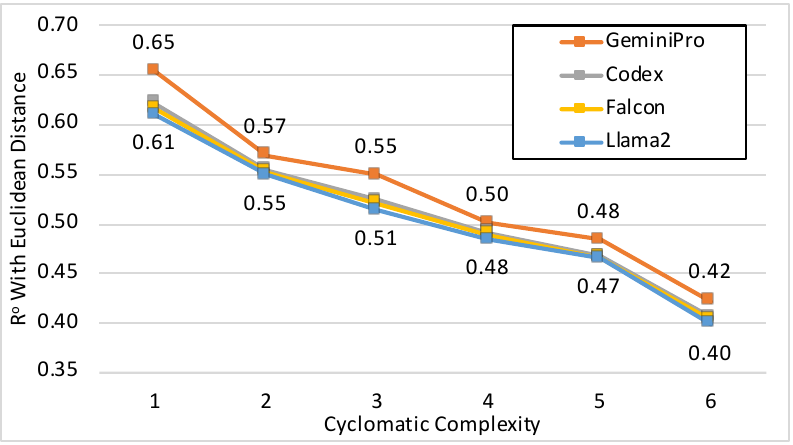} \\ 
\scriptsize{(b) Euclidean Distance} \\
\caption{Variation of Robustness with Cyclomatic Complexity}
\label{fig4}
\end{figure} 

Figure \ref{fig4}(a) and (b) show how robustness varies with the increase of cyclomatic complexity when the similarity is measured by RoBERTa and Euclidean, respectively. In both cases, our experiment data clearly suggested that the robustness decreases when complexity increases as one could expect. This phenomenon is also clearly observed for all other similarity metrics as well. In fact, the Pearson correlation coefficients between cyclomatic complexity and the robustness score $R^o$ are highly correlated; see Table \ref{tab:CorrelationComplexity} for the Pearson correlation coefficients between cyclomatic complexity and robustness $R^o$ with various distance metrics. 

\begin{table}[htbp]
\centering
\caption{Pearson Correlation Coefficients btw Cyclomatic Complexity and Robustness Measure $R^o$}
\label{tab:CorrelationComplexity}
\begin{scriptsize}
\begin{tabular}{|l|c|c|c|c|}
\hline
\textbf{Similarity Metrics} &\textbf{Codex} &\textbf{GeminiPro} &\textbf{Falcon}&\textbf{Llama2}\\ \hline
RoBERTa&0.9924&0.9594&0.9368&0.9404\\ \hline
BLEU&0.8577&0.8860&0.8989&0.9007\\ \hline
ChrF&0.8287&0.7063&0.8950&0.8796\\ \hline
Cosine&0.7661&0.7610&0.7198&0.7706\\ \hline
Meteor&0.8171&0.8057&0.8188&0.8116\\ \hline
RougeL&0.8483&0.8838&0.8477&0.8498\\ \hline
RougeN&0.8243&0.8699&0.7700&0.7962\\ \hline
Euclidean&-0.9795&-0.9869&-0.9873&-0.9849\\ \hline
Levenshtein-Char&-0.7751&-0.8100&-0.8322&-0.8343\\ \hline
Levenshtein-Word&-0.7252&-0.7158&-0.9029&-0.8960\\ \hline
\end{tabular}
\end{scriptsize}
\end{table}

\begin{observation}
The robustness of a LLM in code generation decreases when the complexity of the coding task increases. 
\end{observation}

\begin{observation}
The proposed method can distinguish the robustness of LLMs in code generation on different scenarios defined according to coding task complexity. 
\end{observation}

To evaluate the robustness of code generation in scenarios of different coding topics, six subsets of $D_1$ were constructed by filtering the topic of the coding task on six distinct topics: Array, String, OOP, Regular Expression, Multi-threading, and Data Structure (DS). This enabled us to compare the model's robustness across different topics. 

Figure \ref{fig5} (a) and (b) show how robustness varies with the topics when the similarity is measured by RoBERTa and Euclidean, respectively. It is observed that the model's robustness scores vary across topics. In particular, as we expected, the robustness scores demonstrated that the models' robustness are poorer on the advanced coding tasks like multi-threading, regular expressions, and complex problems such as data structures rather than simpler coding tasks such as array and string processing tasks.

\begin{figure}[h]
\centering
\includegraphics[width=8cm]{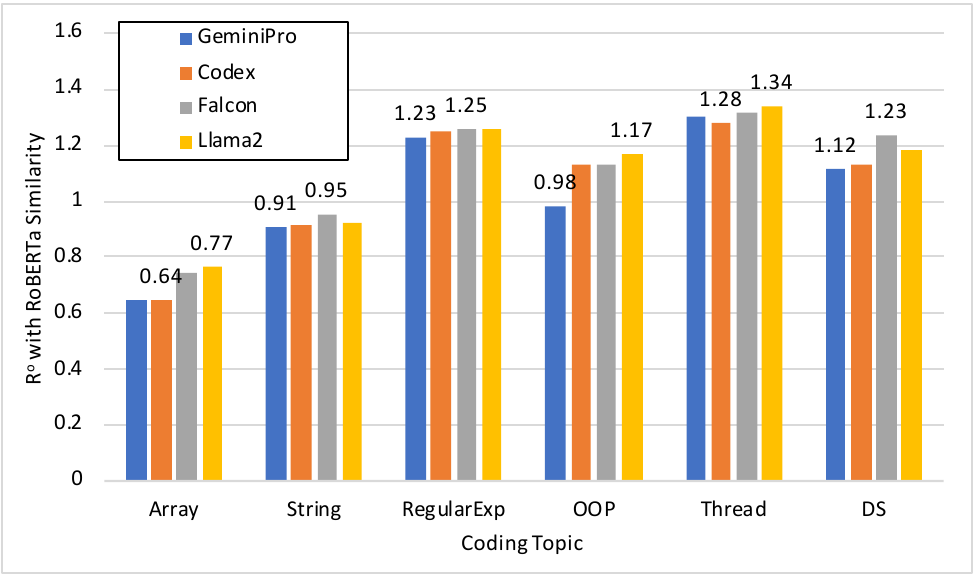}  \\ 
\scriptsize{(a) RoBERTa Similarity} \\
\includegraphics[width=8cm]{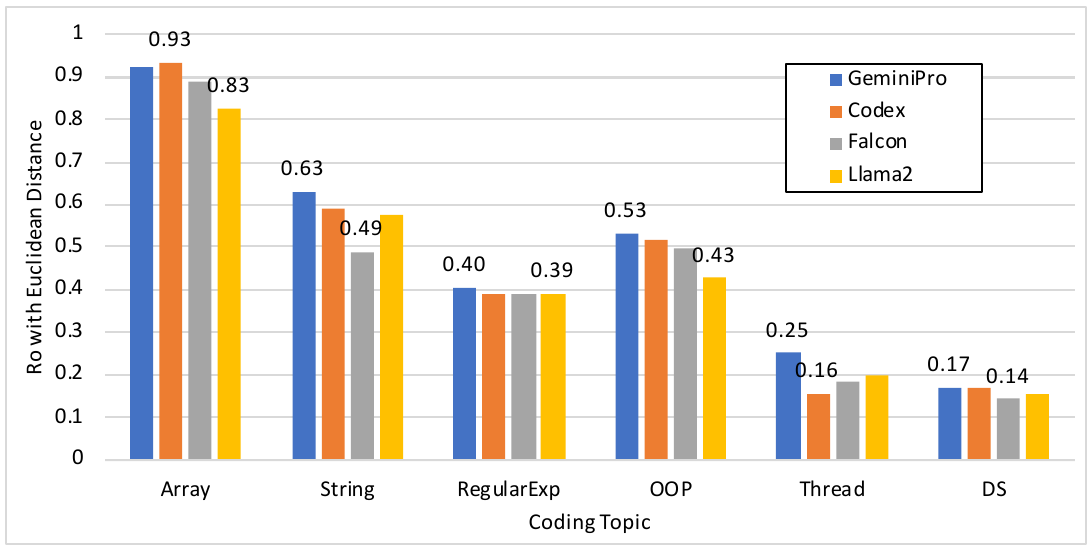} \\ 
\scriptsize{(b) Euclidean Distance} \\
\caption{Variation of Robustness with Topics}
\label{fig5}
\end{figure}

From the experimental data, we can make the following observations. 

\begin{observation} 
The proposed method is capable of distinguishing ML models in their robustness on various scenarios defined by coding topic. 
\end{observation}

\begin{observation}
LLMs performed better in robustness on basic and simple coding tasks than advanced and complicated coding tasks. 
\end{observation}

\subsubsection{Interpretation of Robustness Metrics}

To understand the meanings of the robustness metrics, we analyse the tipping points obtained in the process of testing and evaluation of the robustness. 


In general, a pair consisting of the last success mutant and the first failure mutant is called a \emph{tipping point}. This is where a model switches from success to failure as the ML model moves away from the original test case. In case of code generation, it means that the ML model switches from generating a correct program to generating an incorrect program although the difference between the input coding task descriptions are very small. Therefore, the pair $\left(LS, FF\right)$ of the last success and first failure solutions is a tipping point. 

We now analyse the difference between the programs $LS$ and $FF$ generated by LLMs on such pairs of mutants to answer two research questions: (a) are the programs generated by a LLM before and after a tipping point similar to each other syntactically and semantically? (b) when will a tipping point happen? 

For question (a), we use code editing distance to measure the similarity between two pieces of code by counting the minimum number of editing operations required to transform one piece of code into the other. 

In particular, we use the editing distances based on Abstract Syntax Trees (ASTs) to reflect the syntax distances, and data flow graphs (DFG) and Program Dependency Graphs (PDGs) for more semantics-oriented distances. 

Let $t \in T$ be a seed test case and $m$ be AST, DFG or PDG. We write $Dist^{LS}_{m}(t)$ to denote the distance between the reference solution of the seed test case $t$ and the program generated from the last successful mutant of the seed $t$. Similarly, we write $Dist^{FF}_{m}(t)$ for the distance between the reference solution of the seed test case $t$ and the program generated from the first failed mutant. 
We then calculate the differences between $Dist^{LS}_{m}(t)$ and  $Dist^{FF}_{m}(t)$ as follows. 
\begin{equation}
Diff_m(t) = Dist^{FF}_{m}(t) - Dist^{LS}_{m}(t)
\end{equation}

The statistical results over all test cases in $T$ are shown in Figure \ref{fig6} with the maximum (the upper black bar), minimum (the lower black bar), average (the red dot), and standard deviation (the blue line segment), and the central quartiles (the black box) for three different code editing distance metrics AST, DFG and PDG.  

\begin{figure*}[htb]
\centering
\includegraphics[width=12cm]{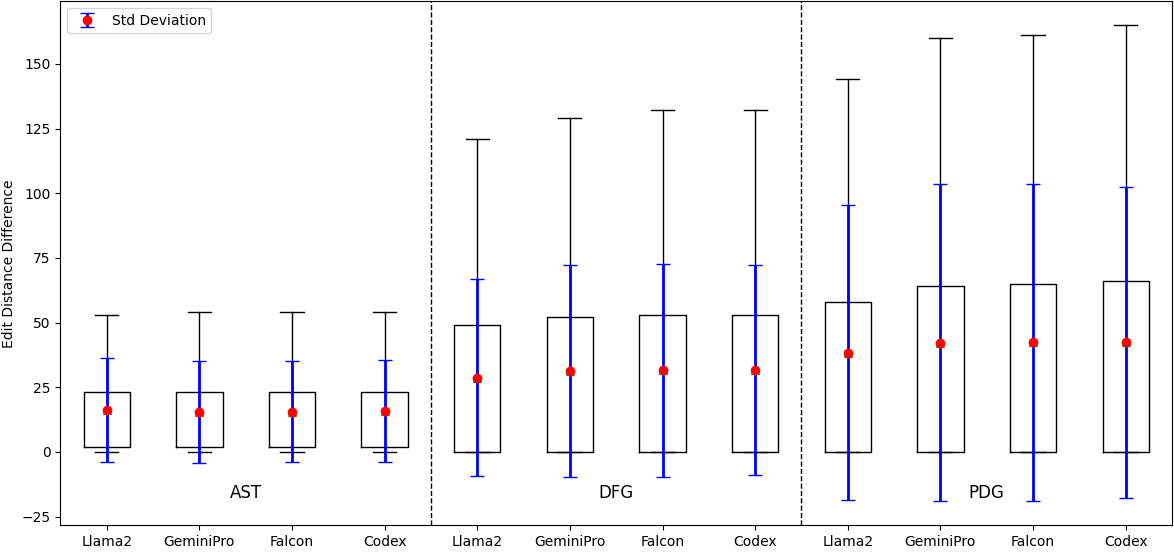} \\ 
\caption{Differences btw the Edit Distances to the Original Description before and after the Tipping Points}
\label{fig6}
\end{figure*}

The data demonstrate that, for all LLMs, the average values of $Diff(t)$ are positive for all LLMs and well above zero. And, most of the data are distributed above the zero. Thus, we have the following observations. 

\begin{observation}
On average, the last successful solution generated by a LLM is closer to the reference solution than the first failed solution when the paraphrase increases its distance to the original coding task description. 
\end{observation}

The distribution of the differences in the distances to the reference solutions before and after tipping points as shown in Figure \ref{fig6} shows that on average there is a big jump in the distances to the reference solution before and after the tipping points. 

\begin{observation}
In most cases, there is a big increase in the distance to the reference solution before and after a tipping point.
\end{observation}

Although the above two observations hold on average, there are exceptional cases and we have made the following observation. 

\begin{observation}
The distance to reference solution may decrease after a tipping point. In other words, it is possible that the first failure solution has a distance closer to the reference solution than the last success solution. 
\end{observation} 

The second analysis of the tipping points aims to provided a comprehensible explanation of the $R^{o}$ and $R^*$ metrics. 
We analysed the paraphrases of the first failed mutants to find out how many words were changed and the semantic similarity between the original words and the replacing words. 

In order to achieve this goal, we calculated the following two values of the paraphrases from each first failed mutants. 
\begin{itemize}
\item $k$: the number of words in the original text description of coding task in the seed test case replaced by different words to generate the paraphrase of the mutant test case. 
\item $n$: the largest distance between a word in the original text description of coding task of the seed and the word used to replace it to generate the paraphrase; here, $n$ is the order of the replacement word in the neighbourhood of the original word in the word embedding vector space.  
\end{itemize}

Table \ref{tab:averageNandK} shows the average values of $n$ and $k$ over all test cases for different metrics used in the robustness evaluation. 

\begin{table}[h]
\centering
\caption{Average Values of $n$ and $k$ of the First Failure Paraphrases}
\label{tab:averageNandK}
\begin{scriptsize}
\begin{tabular}{|l|c|c|c|c|c|c|c|c|}
\hline
\multirow{2}{*}{\textbf{Metric}} &\multicolumn{2}{|c|}{\textbf{Codex}} &\multicolumn{2}{|c|}{\textbf{GeminiPro} }&\multicolumn{2}{|c|}{\textbf{Llama2}}&\multicolumn{2}{|c|}{\textbf{Falcon}}\\ \cline{2-9}
&$n$ &$k$ &$n$ &$k$ &$n$ &$k$ &$n$ &$k$ \\ \hline
RoBERTa&3.67&2.03&3.86&2.73&2.65&2.43&2.67&2.03\\ \hline
ChrF&2.67&2.62&2.68&2.03&2.46&2.11&2.64&2.22\\ \hline
Euclidean&2.76&2.73&2.67&2.48&2.47&2.32&2.68&2.03\\ \hline
Cosine&2.86&2.23&2.67&2.33&2.48&2.41&2.38&2.33\\ \hline
BLEU&2.63&2.09&2.22&2.63&2.63&2.03&2.63&2.04\\ \hline
METEOR&2.68&2.03&2.45&2.43&2.58&2.55&2.62&2.38\\ \hline
Levenshtein-Char&2.57&2.03&2.56&2.03&2.57&2.03&2.57&2.03\\ \hline
Levenshtein-Word&3.60&1.42&3.56&1.33&2.46&1.09&2.42&1.14\\ \hline
ROUGE-N&2.77&2.2&2.99&2.16&2.46&2.11&2.33&2.12\\ \hline
ROUGE-L&2.57&2.32&2.77&2.34&2.41&2.14&2.21&2.03\\ \hline\hline
\textbf{Average}&\textbf{2.88}&\textbf{2.17}&\textbf{2.84}&\textbf{2.25}&\textbf{2.52}&\textbf{2.12}&\textbf{2.52}&\textbf{2.04}\\\hline
\end{tabular}
\end{scriptsize}
\end{table}

From the data presented in Table \ref{tab:averageNandK}, we have the following observation. 

\begin{observation}
Assume that a LLM generates a correct program from a natural language description of the coding task. Then, on average, the LLM will generate an incorrect program code if one changes 2 to 3 words in the description of the coding task by replacing them with words that are the 2nd to 3rd closest to the original words in meaning. If one changes the description smaller than that, the LLM will most probably still be able to generate a correct code. 
\end{observation}

In conclusion, from the above observations, we can conclude that the proposed method is effective in the sense that it is capable of providing meaningful measurements of LLMs that can distinguish different LLMs on robustness in code generation and the robustness of a LLM on various coding scenarios. 

\subsection{Efficiency}

We now address research question \emph{RQ2} on the efficiency of the proposed method. We will assess the feasibility of the proposed method by the computational resource required on computation time, memory space, and the cost of using the LLMs under test. 

\subsubsection{Cost of Using LLMs} 

First, we measure the cost of using the LLM under test by the number of queries to the ML models for each similarity metric. The data on the average numbers of LLM queries for each of the four LLMs and various distance metrics are given in Table \ref{tab:AverageLLMQueries}. 

\begin{table}[ht]
\centering
\caption{Average Number of LLM Queries Per Test Case}
\label{tab:AverageLLMQueries}
\begin{scriptsize}
\begin{tabular}{|l|c|c|c|c|c|c|c|}
\hline
\textbf{Metric} & \textbf{GeminiPro} &\textbf{Codex} &\textbf{Falcon} &\textbf{Llama2}\\ \hline
BLEU &14.82&14.13&13.79&13.65\\ \hline
ChrF &16.02&15.34&14.97&14.85\\ \hline
Levenshtein-Char &16.63&15.96&15.56&15.44\\ \hline
Levenshtein-Word &14.98&14.35&13.97&13.87\\ \hline
ROUGE-N &15.27&14.59&14.23&14.08\\ \hline
ROUGE-L &15.46&14.45&14.09&13.94\\ \hline
METEOR &14.45&13.77&13.42&13.29\\ \hline
RoBERTa &14.80&14.10&13.76&13.62\\ \hline
Cosine &14.99&14.32&13.99&13.85\\ \hline
Euclidean &14.90&14.24&13.91&13.78\\ \hline
\end{tabular}
\end{scriptsize}
\end{table}

From the data presented in Table \ref{tab:AverageLLMQueries}, we observed that, first, in comparison with adversarial robustness evaluation methods, the proposed method is much more efficient and less expensive in terms of the required number of LLM queries. They only require fewer than 17 queries of the LLM on average. Second, there is only a very small variation in the average numbers of queries to the LLMs required by different metrics. So, we have the following conclusion. 

\begin{observation}
From the cost of using the LLM under test point of view, the proposed method is feasible for practical use because it only requires a very modest number of queries to the LLM under test.  
\end{observation}

\subsubsection{Computation Time and Memory Space Consumptions} 

The required computation time for each test case consists of the times for each of the following steps. 
\begin{enumerate}
\item Generating a set of paraphrases. 
\item Calculating the distances between paraphrases and the original test case. 
\item Sorting the paraphrases according to distances.
\item Invoking the LLM to generate code from a paraphrase.  
\item Compiling and testing the generated code to determine if the LLM generated code is correct.
\end{enumerate}

The final two steps are repeated until incorrect code has been generated from a paraphrase by the LLM. 

From the analysis above, a key factor of the computational time required for each test case is the number of paraphrases generated from coding task description. In the experiments, we recorded the number of paraphrases generated from each test case, where the parameters are $n=5$ and $k=5$, which are about twice the average values of $n$ and $k$ required for tipping points; see Table \ref{tab:averageNandK}. 

The experiment results are: the average number of paraphrases over all 900 test cases in the test dataset is 235.79, while the minimum and maximum number of paraphrases are 28 and 1770, respectively. The standard deviation is 269.74. Figure \ref{fig:Paraphrases} shows the distribution of the number of paraphrases generated from test cases. 

\begin{figure}[htb]
\centering
\includegraphics[width=8cm]{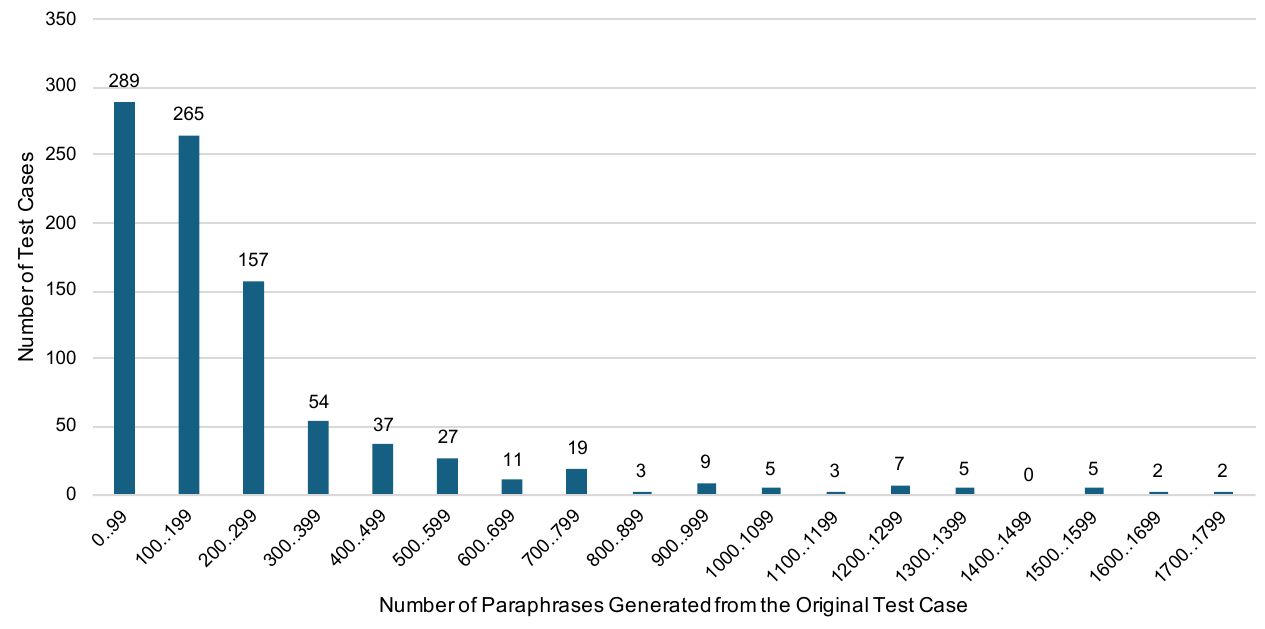} \\ 
\caption{Distribution of the Number of Generated Paraphrases}
\label{fig:Paraphrases}
\end{figure}

Therefore, we have the following observation. 

\begin{observation}
The generation of paraphrases and the sorting of them according to distance from the original coding task description can be very efficient. 
\end{observation}

We have also recorded the CPU time spent on calculating the distances (or similarity scores) between the paraphrases and the original coding task descriptions for each distance metric. Note that, given a distance metric, the time to calculate the distance between two texts may depend on the lengths of the texts. Table \ref{tab:distanceMetricTime} shows the average time required to calculate one distance as well as the minimum, maximum times (in seconds) and the standard deviations.  

\begin{table}[ht]
\centering
\caption{Times to Calculate the Distance bwt The Original Task Description And A Paraphrase}
\label{tab:distanceMetricTime}
\begin{scriptsize}
\begin{tabular}{|l|c|c|c|c|}
\hline
\textbf{Metric} &\textbf{Average} &\textbf{StDev} &\textbf{Max} &\textbf{Min}\\\hline
BLEU&0.2399&0.0078&0.2922&0.2304\\\hline
ChrF&0.2489&0.0078&0.297&0.2394\\\hline
Levenshtein-Char&0.4730&0.0079&0.5229&0.4634\\\hline
Levenshtein-Word&0.4799&0.0079&0.5248&0.4704\\\hline
RougeN&0.8812&0.0082&0.9333&0.8714\\\hline
RougeL&0.9119&0.0078&0.9629&0.9024\\\hline
Meteor&1.1671&0.0082&1.2213&1.1575\\\hline
RoBERTa&3.6302&0.0082&3.6795&3.6205\\\hline
Cosine&3.7341&0.0077&3.7751&3.7244\\\hline
Euclidean&3.9699&0.0077&4.0172&3.9604\\\hline
\end{tabular}
\end{scriptsize}
\end{table}

From Table \ref{tab:distanceMetricTime}, we can observe that, first, the time required to calculate distances varies significantly among different metrics, ranging from 230 ms to nearly 4 seconds per paraphrase. Second, for a given distance metric, the variations on the times required to calculate the distances over the set of test cases is very small. Therefore, depending on the distance metric used, the time to calculate the distances takes from about 1 minute for BLEU and ChrF to about 16 minutes for RoBERTa, Cosine and Euclidean on average per test case. If the distance metric is one of RoBERTa, Cosine and Euclidean, calculating the distances for the set of paraphrases is one of the most time consuming steps. 

The required time for the final step consists of the time taken to query from the LLM, generate test data from the generated code, compile and execute the generated code on the test cases, and determine the correctness of the generated code. Our experiment data shows that on average each query of the LLM may take about 18 seconds, while generation of test data takes about 320 seconds and the execution of the code takes about 16 seconds. In total, the final two steps take about 350 seconds. The generation of test data from the LLM generated code required the most computation time because it is implemented by employing the EvoSuite software testing tool, which is built on an evolutionary computing algorithm. 

The memory space consumed in the testing process consists mostly of the space used to store the set of paraphrases together with their distances to the original text. On average, less than 7MBs of memory space is required for each test case. Note that when the test on one coding task completes, the memory space can be released. Therefore, the demand on memory space is minimal. 

From the above experiment data, we can now make the following conclusion. 

\begin{observation}
The computation time and memory space required to perform robustness test in our proposed method are affordable even when the most time consuming distance metrics are used. 
\end{observation}

It is worth noting that it is still time consuming to conduct a large scale experiment with thousands of test cases. However, experiments can be conducted by executing the algorithm on multiple desktop computers in parallel. 

%

\subsection{Assessment of Text Similarity Metrics}\label{sec:ImpactOfSimilarityMetrics}

There are a wide range of similarity metrics for natural language text that are available to measure the distances between natural language descriptions of coding tasks. Our experiments reported in Section \ref{sec:Experiment} shows that the robustness metrics $R^o$ and $R^*$ vary with the choice of the distance metric although they all provide valid results. Thus, research question \emph{RQ3} arises: which similarity metric is most appropriate for the proposed method? 

From the test process of the proposed method, it is apparent that the accuracy of the robustness metrics $R^o$ and $R^*$ and the usefulness of the evaluation results depend on whether the distance metric on text can separate paraphrases from each other and order them meaningfully according the the distance to the original text. Thus, we introduce the notion of \emph{distinguishability} of distance metrics to guide the selection of distance metric. 

\emph{Distinguishability} is a property of distance metrics on text data space that concerns how well paraphrases can be separated from each other according to their distances to the original text. For a distance metric to be useful in our method, it should be able to distinguish paraphrases according to their distances to the original text. Formally, we have the following definition. 

\begin{definition}(Distinguishability of Distance Metrics)

Let $\|\cdot, \cdot\|$ be a distance metric on texts. Suppose $x$ is an original text, and $x'$ and $x''$ are two paraphrases of $x$. We say that $x'$ and $x''$ are \emph{distinguishable} by the metric $\|\cdot, \cdot\|$ with respect to $x$, if and only if $\|x, x'\| \neq \|x, x''\|$. 
\qed
\end{definition}

However, it is hard to formally derive the distinguishability property from the definition of a distance metric. Therefore, in this paper, we conduct experiments to evaluate the distinguishability of commonly used distance metrics on the text data space. 

We will first define a set of metrics to measure distinguishability, and then report an experiment to demonstrate which similarity/distance metric performs best. 

\subsubsection{Measuring Distinguishability}

The following three metrics for measuring distinguishability are calculated and compared.


\noindent{\emph{A. Uniqueness}} 

Given a distance metric, if a paraphrase $x$ has a unique distances to the original text $t$ among a set $E(t)$ of paraphrases, it can be distinguished from all other paraphrases in $E(t)$ by its distance. The more paraphrases in $E(t)$ have unique distances, the better distinguishability the metric has. Thus, we define \emph{uniqueness score} as the proportion of paraphrases in $E(t)$ that have a unique distance. Formally, let $T$ is the test dataset. 
The \emph{uniqueness score} is calculated as follows. 

\begin{equation}
   Unique(T)= \frac{1}{|T|}\sum_{t \in T} \left(\frac{|\{ x \in E(t) | Unique_{E(t)}(x) \}|}{|E(t)|}\right),
\end{equation}
where 		
%
\[
    Unique_{E(t)}(x) = \forall x' \in E(t) \left(x \neq x' \Rightarrow \|x, t\| \neq \|x', t\| \right)
\]

\noindent{\emph{B. Distinctness}}
		
A distance metric is more capable of distinguishing paraphrases if it can produce more distinct values for the distances on a set of paraphrases. It can provide a good measurement of distinguishability even if the uniqueness score is not high. We calculate the distinctness score of a distance metric using the following formula. 

\begin{equation}
 Distinct(T) = \frac{1}{|T|}\sum_{t \in T}\left(\frac{\left| \{ \|x, t\| ~|~ x \in E(t) \} \right|}{|E(t)|}\right)
\end{equation}
    
\noindent{\emph{C. Differentness}}

When paraphrases are more spread out in terms of their distances to the original text, they are easier to differentiate. We define the average absolute differences of the distances by the following formula. 

\begin{equation}
Difference(T) = \frac{1}{|T|}\sum_{t \in T}\left(\frac{\sum_{x, y \in E(t)} |~ \|x, t\| - \|y, t\| ~|}{|E(t)|^2}\right) \label{equ:DifferenceScore}
\end{equation}

In equation (\ref{equ:DifferenceScore}), the numerator calculates the sum of the absolute differences between the distances of all pairs of paraphrases $x$ and $y$ in the set $E(t)$ from the original text $s$ while the denominator, $|E(t)|^2$, represents the total number of combinations of paraphrases when considering all pairs, accounting for both $(x, y)$ and $(y, x)$.

\subsubsection{Experiment Results} 
In the experiment, the test dataset is the same $D_1$ as in the experiment reported in the previous subsections; see Subsection \ref{sec:designOfExperiment}. For each $t$ in $D_1$, $E(t)$ is the same set of paraphrases used in the experiments reported in the previous experiment. They are paraphrases generated with $n$ varying in the range from a minimum of 1 to a maximum of 5 nearest neighbours, and the value of $k$ varies from 1 up to 11, limited by the length of the coding task description. 

The results are give in Table \ref{tab:uniqueness}. 

\begin{table}[h]
\centering
\caption{Comparison of Distance Metrics}
\label{tab:uniqueness}
\begin{scriptsize}
\begin{tabular}{|l|c|c|c|}
\hline
\textbf{Metric} & \textbf{Uniqueness (\%)} & \textbf{Distinctness} & \textbf{Differentness} \\ \hline
RoBERTa&100.00&1.0000&0.2334 \\ \hline
ChrF&93.22&0.9655&0.2169 \\ \hline
Euclidean&73.94&0.8721&0.2010 \\ \hline
Cosine&35.35&0.6104&0.0789 \\ \hline
BLEU&14.45&0.3242&0.1291 \\ \hline
METEOR&0.00&0.0804&0.1454 \\ \hline
Levenshtein-Char&0.55&0.0661&0.2267 \\ \hline
Levenshtein-Word&0.00&0.0101&0.3054 \\ \hline
ROUGE-N&0.00&0.0452&0.1576 \\ \hline
ROUGE-L&0.00&0.0374&0.0912 \\ \hline
\end{tabular}
\end{scriptsize}
\end{table}

In Table \ref{tab:uniqueness}, the Uniqueness column gives the uniqueness scores of various distance metrics. 
The Distinctness column of Table \ref{tab:uniqueness} gives the distinctness score of various distance metrics. 
Column Differentness gives the differentness score, i.e. the average differences in distances for various distance metrics. Since similarity metrics operate on different scales, we normalized them using min-max normalization to ensure that all metrics are within the range $[0, 1]$. 

From the experiment data given in Table \ref{tab:uniqueness}, the following observations can be made. 

\begin{observation}
RoBERTa distance consistently outperforms all other metrics in uniqueness, distinctness, and difference. Each paraphrase is assigned a unique distance to the original text. The average differences in the distances to the seed over all pairs of paraphrases is around 0.23. 
\end{observation}

\begin{observation}
ChrF and Euclidean also perform reasonably well, though both significantly trail behind RoBERTa especially on uniqueness and distinctness, suggesting they may not capture paraphrase variations as effectively as RoBERTa. 
\end{observation}

\begin{observation}
The overlap based and editing distance metrics performed less satisfactory. Among these metrics, BLEU stands out, showing better scores on uniqueness and distinctness than METEOR, ROUGE-N, ROUGE-L, and Levenshtein-word, which all show 0\% uniqueness and poor distinctness scores.
\end{observation}

\begin{observation}
Cosine has the lowest performance among non-overlap based metrics, but better than overlap based and editing metrics on uniqueness and distinctness, but poorer on differentness. 
\end{observation}

In summary, the semantics based metric \emph{RoBERTa} is the best distance metric for the evaluation of LLM's robustness in code generation. 

\subsection{Threats to Validity}

This section discusses the threats to the validity of our experiments.  

\subsubsection{External Threats} 

Threats to external validity concern the extent to which our findings can be generalized. Our evaluation relies on the ScenEval benchmark, which exclusively consists of Java programming problems, including both text book and real-world Java programming tasks. However, the specific robustness metrics $R^o$ and $R^*$ for each LLM and in various scenarios may not be directly transferable from Java to other programming languages, as different languages have unique features. 
To address this limitation, future work could modify ScenEval to incorporate a broader range of languages or even use a different benchmark. However, despite this, we are confident that the conclusions we have drawn about the effectiveness and efficiency of the proposed method should be generalisable to all programming languages and all scenarios. In addition, the observations on the distinguishability of similarity/distance metrics should be generalisable since they are independent of the LLMs. 

In conclusion, the method proposed in this paper should be generalisable to other LLMs, programming languages and benchmarks. Nevertheless, further experiments to confirm this will be interesting.

\subsubsection{Internal Threats} 

A potential internal threat to the validity of our study is in the selection of test cases. If the test cases used in the evaluation are not sufficiently diverse, this could introduce bias in the results. For instance, over-representation of certain types of coding tasks may lead to misleading conclusions about the robustness of the LLMs. To mitigate this threat, we ensured that the test cases covered a broad range of topics and levels of difficulty, ensuring a comprehensive evaluation. Moreover, we used random sampling where possible to avoid selection bias and to provide a more accurate assessment of the model performance across diverse scenarios. 

Another potential internal threat to the validity of our study is in the choices of the parameters of Algorithm \ref{alg:EvalRobustness} in our experiments, which include $k$ and $n$ for generating the initial set of paraphrases, and $c_k$ and $c_n$ for generating additional paraphrases when an expansion of the set of paraphrases is required. These parameters have impact on the accuracy of the calculations of $R^o$ and $R^*$ metrics. The values of $k$ and $n$ used in our experiments were 5, which is about double of the average values of the tipping points (i.e. the last success and first failure points) over the whole test dataset. It is very rare that expansion is needed. The values of $c_k$ and $c_n$ were 1, which was sufficient to found the tipping points after expansion. 

To estimate the accuracy of the experiment data, we refer to Theorem \ref{thm:Expansion} and \ref{thm:Expansion2} , which formally proved that the real values of $R^o$ and $R^*$ must be between the estimations obtained in an experiment using a set of mutants. Table \ref{tab:Accuracy} gives the absolute differences between the experiment values of $R^o$ and $R^*$, which can be found in Table \ref{tab:combined_robustness}. Column \emph{Ratio} of Table \ref{tab:Accuracy} gives the ratios of the absolute differences between $R^o$ and $R^*$ to the mean values of them, i.e. $Ratio = 2\frac{|R^o-R^*|}{R^o+R^*}$. The data in Table \ref{tab:Accuracy} shows that the overall average of the ratios is less than 5\%. This means the error rate in the experiment data is less than 5\%. Thus, the experiment data obtained from the experiment with the choices of the parameters in Algorithm \ref{alg:EvalRobustness} is highly accurate.  

\begin{table*}[htp]
\caption{Absolute Differences between $R^o$ and $R^*$ and The Ratios to Their Means}\label{tab:Accuracy}
\scriptsize
\begin{center}
\begin{tabular}{|c|c|c|c|c|c|c|c|c||c|c|}
\hline
\multirow{2}{*}{Similarity Metric} &\multicolumn{2}{|c|}{GeminiPro} &\multicolumn{2}{|c|}{Codex} &\multicolumn{2}{|c}{Llama2} &\multicolumn{2}{|c||}{Falcon} &\multicolumn{2}{|c|}{\textbf{Average}}\\\cline{2-11}
 &Difference &Ratio &Difference &Ratio &Difference &Ratio &Difference &Ratio &Difference &Ratio\\\hline
RoBERTa  &0.0752 &0.0700 &0.0553 &0.0496 &0.0533 &0.0436 &0.0504 &0.0410 &0.0586 &0.0511\\\hline
Euclidean  &0.0316 &0.0614 &0.0289 &0.0583 &0.0273 &0.0703 &0.0284 &0.0724 &0.0291 &0.0656\\\hline
Cosine Similarity &0.0011 &0.0012 &0.0080 &0.0086 &0.0010 &0.0010 &0.0009 &0.0009 &0.0027 &0.0029\\\hline
BLEU  &0.0292 &0.0661 &0.0224 &0.0524 &0.0292 &0.0533 &0.0288 &0.0530 &0.0274 &0.0562\\\hline
ROUGE-N  &0.0204 &0.0322 &0.0233 &0.0360 &0.0248 &0.0329 &0.0225 &0.0300 &0.0228 &0.0328\\\hline
ROUGE-L  &0.0180 &0.0256 &0.0169 &0.0236 &0.0173 &0.0211 &0.0157 &0.0192 &0.0170 &0.0224\\\hline
METEOR  &0.0233 &0.0349 &0.0233 &0.0342 &0.0242 &0.0307 &0.0221 &0.0282 &0.0232 &0.0320\\\hline
ChrF  &1.8729 &0.0276 &1.7555 &0.0255 &1.6484 &0.0237 &1.7166 &0.0248 &1.7484 &0.0254\\\hline
Levenshtein-Word  &0.1110 &0.1029 &0.1137 &0.1061 &0.1164 &0.1090 &0.0697 &0.0639 &0.1027 &0.0955\\\hline
Levenshtein-Char  &1.4018 &0.1296 &1.0976 &0.1093 &0.9640 &0.0999 &1.0320 &0.1051 &1.1239 &0.1110\\\hline\hline
\textbf{Average} &\textbf{0.3585} &\textbf{0.0551} &\textbf{0.3145} &\textbf{0.0504} &\textbf{0.2906} &\textbf{0.0486} &\textbf{0.2987} &\textbf{0.0439} &\textbf{0.3156} &\textbf{\textit{0.0495}}\\\hline
\end{tabular}
\end{center}
\label{default}
\end{table*}%

\subsubsection{Construct Threats}

Construct validity concerns the appropriateness of the measures and metrics used in the experiments. It manifests itself in the way we determine the functional equivalence between the codes generated from the original coding task and from its paraphrases, and in the metrics used to order the mutants for testing. 

We considered two programs to be functionally equivalent if they pass all the test cases generated from them. A potential threat to construct validity is the possibility that the test suite may be inadequate and that the differences in the generated programs might be overlooked. To mitigate this threat, we implement the \(generateTestCode\) test morphism with the Evosuite tool, as detailed in \citep{ghosh2024}. Evosuite is known for its strong test adequacy \citep{fraser2014, fraser2011}, ensuring high test coverage in the generated test cases. This is also proved empirically by \citet{ghosh2024}. In addition, the test cases are generated automatically from both programs generated by the LLM on the original and paraphrased coding tasks. Therefore, they can detect omission and commission errors. Consequently, there is minimal risk that the results of the experiment can have systematic bias due to omission of critical program bugs. 

The other threat to validity is that the wrong distance metric was used to order the mutants. Therefore we studied a wide range of metrics. These were overlap-based metrics BLEU, ROUGE, Meteor, and ChrF, edit distance metrics Levenshtein distance both at word and character levels, vector space distance metrics Euclidean and Cosine similarity, as well as the semantics-based metric based on the \emph{STSb-RoBERTa-base} ML model. For all of these metrics, the experiments demonstrated that the proposed method can provide meaningful evaluations that are capable of distinguishing LLMs based on their robustness and assessing the robustness in different scenarios. The experiments also demonstrated the feasibility of the proposed method. While each metric helps to reveal different aspects of the LLM's robustness on code generation, the impact of the choice of distance metrics is minimal. Thus, the risk associated with the choice is minimal. 

\section{Conclusion}\label{sec:Conclusion}

We proposed a new method to evaluate the operational robustness of LLMs for the task of code generation. It is the first micro method for evaluating operational robustness of NLP on generative tasks. Our work demonstrates the effectiveness of the method. Our experimental results offer valuable insights into the robustness of LLMs on code generation. In particular, in terms of the similarity between the original coding task description and the paraphrases, we found that an LLM, like Codex, will start to generate incorrect programs from the paraphrase when the paraphrase similarity is 0.42 for BLEU,  0.93 for Cosine, and 0.66 for Meteor, for example. This means, roughly speaking, on average a state-of-the-art LLM will start to generate incorrect programs if the task description is modified by replacing 2 to 3 words in the text with words of similar meanings (the 3 to 4 closest ones in word embedding vector space). 

Through experimentation, we have demonstrated that the proposed method is both feasible and cost-effective for evaluating the robustness of LLM models. It requires fewer than 20 queries to the LLM per test case and the required computation resources in terms of CPU time and memory are very modest and affordable. 

We have implemented a test system that automates the robustness evaluation process in our proposed method with the support of the Morphy tool. 

For future work, we aim to expand our study of LLM for code generation to other aspects of usability. While our current evaluation focuses on robustness, we recognize the importance of assessing other critical aspects of code quality, including various types of code smells. By analysing these factors, we intend to provide a more comprehensive evaluation of the performance of LLMs in code generation, ultimately leading to insights that can enhance the usability and effectiveness of code generation in real-world applications.

The proposed method should also be applied to other types of ML tasks on their micro operational robustness. It will be very interesting to see such experiments and case studies.



\bibliographystyle{elsarticle-harv} 
\bibliography{references}


\newpage

\section*{Appendix. Proofs of Theorems} 

\proof{Theorem \ref{thm:NablaDelta}}  

We prove the theorem by contradiction. Assume the theorem is not true; that is, there is $\alpha \in D^s$ such that 
\begin{equation}
\nabla_{D^s}(t,M) < \|\alpha, t\| < \Delta_{D^s}(t,M). \label{enq:thm0eqn1}
\end{equation}

We have that either $Fail^M_t(\alpha)$ or $\neg Fail^M_t(\alpha)$. The following proves that in both situations, the assumption (\ref{enq:thm0eqn1}) is not true. 

\emph{Situation} (1): Assume that 
\begin{equation}
Fail^M_t(\alpha). \label{eqn:assumption1}
\end{equation}

By Lemma \ref{lmm:lmmB}, we have that $\|\alpha,t\| \geq \Delta_{D^s}(t,M)$. This contradicts (\ref{enq:thm0eqn1}). Therefore, the assumption (\ref{eqn:assumption1}) is not true. 

\emph{Situation} (2): Assume that 
\begin{equation}
\neg Fail^M_t(\alpha). \label{eqn:assumption2}
\end{equation}
By (\ref{enq:thm0eqn1}), i.e. $\|\alpha, t\| > \nabla_{D^s}(t,M)$, and Lemma \ref{lmm:lmmE}, we have that 
\begin{equation}
\neg Safe^M_t(\alpha). \label{eqn:thm0eqn1}
\end{equation}
By Definition \ref{def:Safe}, (\ref{eqn:assumption2}) and (\ref{eqn:thm0eqn1}) imply that there is $x \in D^s$ that $\|x,t\| < \|\alpha, t\| \wedge Fail^M_t(x)$. Since, by (\ref{enq:thm0eqn1}), $\|\alpha, t\| < \Delta_{D^s}(x,M)$, this implies that $\|x,t\| < \Delta_{D^s}(t,M) \wedge Fail^M_t(x)$. It contradicts Lemma \ref{lmm:lmmB}. Therefore, assumption (\ref{eqn:assumption2}) is not true. 

Because both (\ref{eqn:assumption1}) and (\ref{eqn:assumption2}) are not true, we have that the assumption about the existence of $\alpha$ that satisfy (\ref{enq:thm0eqn1}) is not true. Thus, the theorem statement holds. 
\qed

\proof{Lemma \ref{lmm:lmm1}}
Let $x \in Mut(t)$ and $Fail^M_t(x)$. If $x \prec_t FF(t)$, $x$ will be ordered before $FF(t)$. $Fail^M_t(x)$ contradicts that $FF(t)$ is the first failure. 
\qed

\proof{Lemma \ref{lmm:llm2}}
By definition of last success and first failure, $LS(t)$ is always in the order before $FF(t)$. 
\qed

\proof{Lemma \ref{lmm:lmm3}}

Assume that $\alpha \in Mut(t)$ and 
\begin{equation}
LS(t) \prec_t \alpha \prec_t FF(t). \label{eqn:lmm2Proof1}
\end{equation}

If $Fail^M_t(\alpha)$, by Lemma \ref{lmm:lmm1}, we have that $FF(t) \preccurlyeq \alpha$. This contradicts (\ref{eqn:lmm2Proof1}). 

If $\neg Fail^M_t(\alpha)$, (\ref{eqn:lmm2Proof1}) contradicts that $LS(t)$ is the last success. 

Therefore, $\alpha \in Mut(t)$ is not true. Thus, the Lemma statement holds. 
\qed

\proof{Theorem \ref{thm:Expansion2}}

\noindent{\emph{Proof of Statement (1):}}

Because $\alpha^o$ is the first failure of $Mut(t)$, we have that $\alpha^o \in Mut(t) \subset Mut'(t)$ and $Fail^M_t(\alpha)$. By Lemma \ref{lmm:lmm1}, $\beta^o$ is the first failure of $Mut'(t)$ implies that $\beta^o \preccurlyeq_t \alpha^o$. 

\noindent{\emph{Proof of Statement (2):}}

From the condition $\|\alpha^*, t\| \leq \nabla_{D^s}(t,M)$ on $\alpha^*$, we have that 
\begin{equation}
\forall x \in D^s.\left(x \prec_t \alpha^* \Rightarrow Safe^M_t(x) \right). \label{eqn:Expansion2Proof1}
\end{equation}
By the validity of $Mut('(t)$, (\ref{eqn:Expansion2Proof1}) implies that 
\begin{equation}
\forall x \in Mut'(t).\left(x \prec_t \alpha^* \Rightarrow Safe^M_t(x) \right). \label{eqn:Expanions2Proof2}
\end{equation}
By Definition \ref{def:Safe}, (\ref{eqn:Expanions2Proof2}) implies that 
\begin{equation}
\forall x \in Mut'(t).\left(x \prec_t \alpha^* \Rightarrow \neg Fail^M_t(x) \right).\label{eqn:Expansion2Proof3}
\end{equation}

Because $\beta^*$ is the last success found in $Mut'(t)$, $\beta^*$ is the largest $\tau$ such that 
\begin{equation}
\neg Fail^M_t(\tau) \wedge \forall x \in Mut'(t).\left(x \prec_t \tau \Rightarrow \neg Fail^M_t(x) \right).\label{eqn:Expansion2Proof4}
\end{equation}
In other words, if $\tau$ satisfies (\ref{eqn:Expansion2Proof4}), then, $\tau \preccurlyeq \beta^*$. 

Because $\alpha^*$ is the last success found in $Mut(t)$, we have that $\neg Fail^M_t(\alpha^*)$ and $\alpha^* \in Mut(t) \subset Mut'(t)$. Therefore, (\ref{eqn:Expansion2Proof3}) implies that $\alpha^*$ satisfies the condition (\ref{eqn:Expansion2Proof4}). Thus, $\alpha^* \preccurlyeq_t \beta^*$. 
\qed

\proof{Theorem \ref{thm:Expansion}} 

\noindent{\emph{Proof of Statement} (1):} 

Since $\alpha^o$ is the first failure in the set $\left(Mut'(t)-Mut(t)\right)$, we have that, 
\begin{equation}
Fail^M_t(\alpha^o). \label{eqn:thm3Proof3}
\end{equation}
And, by (\ref{eqn:thm3Proof2}), we have that
\begin{equation}
\alpha^o \in Mut'(t). \label{eqn:thm3Proof3b}
\end{equation}
(\ref{eqn:thm3Proof3}) and (\ref{eqn:thm3Proof3b}) imply that $\alpha^o$ is a failure point in $Mut'(t)$. Because $\beta^o$ is the first failure of $Mut'(t)$, by Lemma \ref{lmm:lmm1}, we have that 
\begin{equation}
\beta^o \preccurlyeq_t \alpha^o. \label{equ:thm3Proof3c}
\end{equation}
On the other hand, because $\alpha^o$ is the first failure in the set $\left(Mut'(t)-Mut(t)\right)$, we have that
\begin{equation}
\forall x \in (Mut'(t)-Mut(t)).\left((x \prec_t \alpha^o) \Rightarrow \neg Fail^M_t(x) \right). \label{eqn:thm3Proof4}
\end{equation}
By (\ref{eqn:thm3Proof2}) and (\ref{eqn:thm3Proof4}), because 
\[Mut'(t) = (Mut'(t) - Mut(t))\cup Mut(t),\] 
we have that 
\begin{equation}
\forall x \in Mut'(t).\left(x \prec_t \alpha^o \Rightarrow \neg Fail^M_t(x) \right). \label{eqn:thm3Proof5}
\end{equation}
Because $\beta^o$ is the first failure of $Mut'(t)$, by the definition of first failure, we have that 
\begin{equation}
\beta^o \in Mut'(t) \wedge Fail^M_t(\beta^o). \label{eqn:thm3Proof6}
\end{equation}
(\ref{eqn:thm3Proof5}) and (\ref{eqn:thm3Proof6}) imply that $\beta^o \prec_t \alpha^o$ is not true. Thus, 
\begin{equation}
\alpha^o \preccurlyeq_t \beta^o. \label{eqn:thm3Proof7}
\end{equation}
(\ref{equ:thm3Proof3c}) and (\ref{eqn:thm3Proof7}) imply that $\alpha^o \approx_t \beta^o$. That is, theorem statement (1) is true. 

\noindent{\emph{Proof of statement} (2):}

By Theorem statement (1), we have that $\alpha^o \approx_t \beta^o$. By Lemma \ref{lmm:lmm3}, we have that $\beta^* \preccurlyeq_t \beta^o \approx \alpha^o \approx_t \alpha^*$. 

By Lemma \ref{lmm:lmm3}, there is no $x \in Mut'(t)$ such that $\beta^* \prec_t x \prec_t \beta^o$. Since $\beta^o \approx_t \alpha^o \approx_t \alpha^*$, the theorem statement (2) is true. 

\noindent{\emph{Proof of statement} (3):}

By Lemma \ref{lmm:llm2}, we have that $\beta^* \preccurlyeq_t \beta^o$. Thus, we prove the statement in two different situations: (a)  $\beta^* \approx_t \beta^o$, (b) $\beta^* \prec_t \beta^o$. 

In situation (a), we have that 
\begin{eqnarray*}
&\alpha^* \prec_t \alpha^o  &(Condition ~of ~Statement ~2)\\
&\approx_t \beta^o &(Theorem ~Statement ~(1))\\
&\approx_t \beta^* &(Condition ~of ~the ~situation)
\end{eqnarray*}

Therefore, we have that $\alpha^* \prec_t \beta^*$. 

By Lemma \ref{lmm:lmm3}, we have that there is no $x \in (Mut'(t) - Mut(t))$ such that $\alpha^* \prec_t x \prec_t \alpha^o$. Since $\alpha^o \approx_t \beta^*$, have that 
\begin{equation}
\neg \exists x \in (Mut'(t) - Mut(t)).\left( \alpha^* \prec_t x \prec_t \beta^* \right). \label{eqn:Thm4Proof2}
\end{equation}
By (\ref{eqn:Thm4Proof2}), we have that 
\begin{equation}
\forall x \in Mut'(t).\left( (\alpha^* \prec_t x \prec_t \beta^*)  \Rightarrow (x \in Mut(t)) \right). 
\end{equation}

In situation (b), because $\alpha^* \in Mut'(t) - Mut(t)$, we have that 
\begin{equation}
\alpha^* \in Mut'(t).\label{eqn:Thm4Proof3}
\end{equation}
Assume that $\beta^* \prec_t \alpha^*$. Then, we have that 
\begin{eqnarray*}
&\beta^* \prec_t \alpha^* &(Assumption) \nonumber\\
&\prec_t \alpha^o &(Condition ~of ~statement) \nonumber\\
&\approx_t \beta^o &(Theorem ~statement ~(1)). 
\end{eqnarray*}

Thus, we have that 
\begin{equation}
\beta^* \prec_t \alpha^* \prec_t \beta^o. \label{eqn:thm4Proof4}
\end{equation}
(\ref{eqn:Thm4Proof3}) and (\ref{eqn:thm4Proof4}) contradict Lemma \ref{lmm:lmm3}, which states that 
\[\neg \exists x \in Mut'(t).\left( \beta^* \prec_t x \prec_t \beta^o \right).\] 
Therefore, the assumption $\beta^* \prec_t \alpha^*$ is not true. Hence, we have that $\alpha^* \preccurlyeq_t \beta^*$. 

By the condition of the situation and theorem statement (1), we have that $\beta^* \prec_t \beta^o \approx_t \alpha^o$. Therefore, 
\begin{equation}
\forall x \in Mut'(t).\left((\alpha^* \prec_t x \prec_t \beta^*) \Rightarrow (\alpha^* \prec_t x \prec_t \alpha^o)\right). \label{eqn:thm4Proof5}
\end{equation}

By Lemma \ref{lmm:lmm3}, we have that 
\begin{equation}
\neg \exists x \in Mut'(t)-Mut(t).\left( \alpha^* \prec_t x \prec_t \alpha^o \right). \label{eqn:thm4Proof6}
\end{equation}
From (\ref{eqn:thm4Proof5}) and (\ref{eqn:thm4Proof6}), we have that 
\begin{equation}
\forall x \in Mut'(t).\left((\alpha^* \prec_t x \prec_t \beta^*) \Rightarrow (x \in Mut(t)) \right). 
\end{equation}

In summary, the statement is true in both situations.  
\qed

\proof{Theorem \ref{thm:Ro}}

By the definition validity and the fact that $FF(t) \in Mut(t)$, we have that 
\begin{equation}
FF(t) \in D^s. \label{eqn:thm1Proof1}
\end{equation}

By Definition \ref{def:MinimalDistanceToFailure} of $\Delta_{D^s}(t,M)$, the fact that $Fail^M_t(FF(t))$ and (\ref{eqn:thm1Proof1}) imply that 
\begin{equation}
\|t, FF(t)\| \geq \Delta_{D^s}(t,M). \label{eqn:thm1Proof2}
\end{equation}

By the definition of completeness, we have that 
\begin{equation}
\forall x \in D^s.\left( \|x,t\|\leq \delta_t \Rightarrow x \in Mut(t)\right), \label{eqn:thm1Proof2b}
\end{equation}
where $\delta_t$ is the coverage size of $Mut(t)$. 
Since $FF(t)$ is the first failure, by Lemma \ref{lmm:lmm1}, we have that 
\begin{equation}
\forall x \in Mut. \left(x \prec_t FF(t) \Rightarrow \neg Fail^M_t(x)\right). \label{eqn:thm1Proof3}
\end{equation}
When the coverage size $\delta_t \geq \|t,FF(t)\|$, (\ref{eqn:thm1Proof2b}) and (\ref{eqn:thm1Proof3}) imply that 
\begin{equation}
\forall x \in D^s.\left(x \prec_t FF(t) \Rightarrow \neg Fail^M_t(x)\right). \label{eqn:thm1Proof4}
\end{equation}
Therefore, by Definition \ref{def:MinimalDistanceToFailure}, we have that 
\begin{equation}
 \Delta_{D^s}(t,M) \geq \|t, FF(t)\|. \label{eqn:thm1Proof5}
\end{equation}
From (\ref{eqn:thm1Proof2}) and (\ref{eqn:thm1Proof5}), we have that $\|t,FF(t)\| = \Delta_{D^s}(t,M)$. 
\qed

\proof{Theorem \ref{thm:R*}}

We first prove that $LS(t) \preccurlyeq_t \nabla_{D^s}(t,M)$. 

By the definition of validity, we have that 
\begin{equation}
LS(t) \in D^s. \label{eqn:thm2Proof1}
\end{equation}
By the definition of $LS(t)$, we have that 
\begin{equation}
\neg Fail^M_x(LS(t)), \label{eqn:thm2Proof2}
\end{equation}
and
\begin{equation}
\forall x \in Mut(t).\left(x \prec_t LS(t) \Rightarrow \neg Fail^M_t(x)\right). \label{eqn:thm2Proof3}
\end{equation}
By the completeness of $Mut(t)$,  we have that 
\begin{equation}
\forall x \in D^s.\left(\|x, t\| \leq \delta \Rightarrow x \in Mut(t) \right). \label{eqn:thm2Proof4}
\end{equation}
where $\delta$ is the coverage size of $Mut(t)$. 
From the fact that $\delta \geq \|t, LS(t)\|$, (\ref{eqn:thm2Proof3}) and (\ref{eqn:thm2Proof4}) imply that 
\begin{equation}
\forall x \in D^s. \left(x \prec_t LS(t) \Rightarrow \neg Fail^M_t(x)\right). \label{eqn:thm2Proof5}
\end{equation}

Definition \ref{def:SafeDistance}, (\ref{eqn:thm2Proof2}) and (\ref{eqn:thm2Proof5}) imply that $LS(t)$ is safe. Formally, 
\begin{equation}
Safe^M_x(LS(t)). \label{eqn:thm2Proof6a}
\end{equation}
Hence, by Lemma \ref{lmm:lmmD}, we have that
\begin{equation}
\|t,LS(t)\| \leq \nabla_{D^s}(t,M). \label{eqn:thm2Proof6}
\end{equation}

Now, we prove that there is no $x \in D^s$ such that $\|t, LS(t)\|\prec_t \|x, t\| \prec_t \nabla_{D^s}(t,M)$. 

By (\ref{eqn:thm2Proof6}), there are two possible situations: (a) $\|t,LS(t)\| = \nabla_{D^s}(t,M)$, (b) $\|t,LS(t)\| < \nabla_{D^s}(t,M)$. 

The theorem statement is trivially true in situation (a). 

In situation (b), 
by Definition \ref{def:SafeDistance} of $\nabla_{D^s}(t,M)$, the condition of the situation implies that 
there are $x \in D^s$ such that $x$ is safe and $LS(t) \prec_t x$. Formally, 
\begin{equation}
\exists x \in D^s.\left( Safe^M_t(x) \wedge (LS(t) \prec_t x ) \right).  \label{eqn:thm2Proof7}
\end{equation}

Let $\alpha$ be any element in $D^s$ that satisfies the predicate of (\ref{eqn:thm2Proof7}). By Lemma \ref{lmm:lmm3}, there is no $x \in Mut(t)$ such that $LS(t) \prec_t x \prec_t FF(t)$. By the completeness of $Mut(t)$ with a coverage size $\delta \geq \|t,FF(t)\|$, we have that there is no $x \in D^s$ such that $LS(t) \prec_t x \prec_t FF(t)$. Thus,  (\ref{eqn:thm2Proof7}) implies that 
\begin{equation}
FF(t) \preccurlyeq_t \alpha. \label{eqn:thm2Proof8}
\end{equation}
We now prove that it is impossible that $FF(t) \prec_t \alpha$. Otherwise, by the definition of $FF(t)$, we have that 
\begin{equation}
Fail^M_t(FF(t)). \label{eqn:thm2Proof9}
\end{equation}
By Definition \ref{def:SafeDistance}, (\ref{eqn:thm2Proof8}), (\ref{eqn:thm2Proof9}) would imply that $\alpha$ is not safe; formally 
\begin{equation}
FF(t) \prec_t \alpha \Rightarrow \neg Safe^M_t(\alpha).\label{eqn:thm2Proof9b}
\end{equation} 
This contradicts the assumption $Safe^M_t(\alpha)$. Thus, we have that $\alpha \preccurlyeq_t FF(t)$. 
Therefore, by Definition \ref{def:SafeDistance}, we have that
\begin{equation}
\left(\nabla_{D^s}(t,M)>\|t,LS(t)\|\right) \Rightarrow \left(\nabla_{D^s}(t,M) = \|t, FF(t)\|\right). \label{eqn:thm2Proof10b}
\end{equation}
By Lemma \ref{lmm:lmm3}, (\ref{eqn:thm2Proof10b}) implies that 
\begin{equation}
\neg \exists x \in Mut(t).\left( LS(t) \prec_t x \prec_t FS(t) \right). \label{eqn:thm2Proof11}
\end{equation}
By the validity and completeness of $Mut(t)$ with coverage size $\delta \geq \|t, FS(t)$, we have that 
\begin{equation}
\forall x \in D.\left( \|x, t\| leq \|FS(x),t\| \Rightarrow (x \in D^s \Leftrightarrow x \in Mut(t)) \right). 
\label{eqn:thm2Proof12}
\end{equation}
(\ref{eqn:thm2Proof11}) and (\ref{eqn:thm2Proof12}) imply that
\begin{equation}
\neg \exists x \in D^s.\left( LS(t) \prec_t x \prec_t FS(t) \right). 
\end{equation}
Hence, theorem statement is true. 
\qed
\end{document}